\RequirePackage[l2tabu, orthodox]{nag}		
\documentclass[10pt]{article}

\usepackage[T1]{fontenc}
\usepackage[frenchmath]{mathastext}							

\usepackage{amsmath}                        
\numberwithin{equation}{section}
\usepackage{graphicx} 											
\usepackage{adjustbox,float}
\graphicspath{{Figs/}}
\usepackage{enumitem}	
\usepackage{mdwlist}												
\usepackage[dvipsnames]{xcolor}							
\usepackage[plainpages=false, pdfpagelabels]{hyperref} 
\hypersetup{
colorlinks   = true,
citecolor    = RoyalBlue, 
linkcolor    = RubineRed, 
urlcolor     = Turquoise
}

\usepackage{amssymb}                        
\usepackage[mathscr]{eucal}                 
\usepackage{dsfont}													
\usepackage{mathtools}                      
\mathtoolsset{showonlyrefs=true}          
\usepackage{amsthm}                         
\allowdisplaybreaks                       
\usepackage{subcaption}
\usepackage{multicol,multirow}

\theoremstyle{plain}
\newtheorem{theorem}{Theorem}
\numberwithin{theorem}{section}
\newtheorem{lemma}{Lemma}       	
\numberwithin{lemma}{section}
\newtheorem{proposition}{Proposition}
\numberwithin{proposition}{section}

\numberwithin{corollary}{section}
\theoremstyle{definition}
\newtheorem{definition}{Definition}
\numberwithin{definition}{section}
\newtheorem{example}{Example}
\numberwithin{example}{section}
\newtheorem{remark}{Remark}
\numberwithin{remark}{section}
\newtheorem{assumption}{Assumption}
\numberwithin{assumption}{section}

\usepackage[
paperwidth=8.5in,
paperheight=11in,
top=1.05in,
bottom=1.05in,
left=1.1in,
right=1.1in,
includehead, 
includefoot
]{geometry}
\usepackage{algorithm}
\usepackage[noend]{algpseudocode}

\newcommand\Eb{\mathds{E}}
\newcommand\Fb{\mathds{F}}
\newcommand\Ib{\mathds{1}}
\newcommand\Gb{\mathds{G}}
\newcommand\Pb{\mathds{P}}

\newcommand\Rb{\mathds{R}}

\newcommand\Ac{\mathcal{A}}
\newcommand\Bc{\mathscr{B}}
\newcommand\Cc{\mathcal{C}}

\newcommand\Fc{\mathscr{F}}
\newcommand\Gc{\mathcal{G}}

\newcommand\Mc{\mathcal{M}}

\newcommand\Qc{\mathcal{Q}}

\newcommand\Ic{\mathcal{I}}

\newcommand\Uc{\mathcal{U}}

\newcommand\om{\omega}
\newcommand\Om{\Omega}
\newcommand\sig{\sigma}

\newcommand\gam{\gamma}
\newcommand\Gam{\Gamma}
\newcommand\lam{\lambda}
\newcommand\del{\delta}

\newcommand\Pho{\overline{\Phi}}

\newcommand\It{\tilde{I}}

\newcommand\Xt{\widetilde{X}}

\newcommand\Ut{\widetilde{U}}

\newcommand\dd{\mathrm{d}}
\newcommand\ee{\mathrm{e}}

\newcommand{\ks}{{k^*}}

\newcommand{\CMIM}{\mathcal{CMIM}}

\newcommand{\diag}{\operatorname{diag}}

\newcommand{\erf}{\operatorname{erf}}

\newcommand \al {\alpha}

\providecommand{\keywords}[1]{\noindent\textbf{\textit{Keywords: }} #1}

\usepackage[
backend=bibtex,
style=ieee-alphabetic, 
sorting=nyt,
giveninits=true,
]{biblatex}
\AtBeginBibliography{\footnotesize} 

\begin{document}

\title{
Rank-Dependent Predictable Forward Performance Processes 
}

\author{
Bahman Angoshtari\thanks{Department of Mathematics, University of Miami, Coral Gables, FL.  \textbf{e-mail}: \url{bangoshtari@miami.edu}}
\and Shida Duan\thanks{Department of Mathematics, University of Miami, Coral Gables, FL.  \textbf{e-mail}: \url{shidaduan@miami.edu}}
}
\date{This version: \today}

\maketitle

\begin{abstract}
Predictable forward performance processes (PFPPs) are stochastic optimal control frameworks for an
agent who controls a randomly evolving system but can only prescribe the system dynamics
for a short period ahead. This is a common scenario in 
which a controlling agent frequently re-calibrates her model. We introduce a new class of PFPPs based on rank-dependent utility, generalizing existing models that are based on expected utility theory (EUT). We establish existence of rank-dependent PFPPs under a conditionally complete market and exogenous probability distortion functions which are updated periodically. We show that their construction reduces to solving an integral equation that generalizes the integral equation obtained under EUT in previous studies. We then propose a new approach for solving the integral equation via theory of Volterra equations. We illustrate our result in the special case of conditionally complete Black-Scholes model.
\end{abstract}



%
%

\section{Introduction}\label{sec:intro}

In the classical models of portfolio choice, model parameters are assumed to be known and the main focus is on characterizing the investment policy that optimizes an objective over an investment horizon. This approach is justified by the separation principle in the stochastic control theory, which states that filtering and control problems decouple under certain conditions. In such settings, it is justified to estimate the model parameters once at the beginning of the investment horizon, obtain the optimal policy implied by the estimated parameters, and follow this policy till the end of the investment horizon.

This is not done in practice, however. Investment managers frequently update their models within their investment horizon. In other words, their investment horizon consists of a sequence of trading periods. At the beginning of each period, new model parameters are estimated which yield the trading policy to be followed in that period until the next model update. The reason for this approach is simple: there is no single model (that is, a model and specific values of its parameters) that describes the underlying market over the whole investment horizon. Each model is trusted for a limited time and becomes obsolete before the end of the investment horizon.

This mismatch (between the investment horizon and how long individual models are valid) leads to a challenge in setting up a mathematical model for obtaining trading policies. On one hand, the investment objectives should be chosen over the longer investment horizon. On the other, models are fully known in advance over shorter periods. We are then left with the following four alternatives:
\begin{itemize}
\item We choose a myopic objective. For such an objective, optimizing over shorter trading periods is equivalent to optimizing over the longer investment horizon. There are not many such objectives, however. One example is to maximize expected logarithm under certain market models (say, Black-Scholes model). Most objectives are non-myopic.

\item We forgo the need for specifying an objective over the whole investment horizon and, instead, treat each trading period as a separate investment horizon. That is, our future planning does not go beyond the time we trust our model. This can be very restrictive, especially if the underlying model needs frequent updating (think of a scenario in which the model is updated daily, and one cannot plan even two days ahead).

\item We try to model how the parameters evolve over the investment horizon. In general, this is a daunting task. After all, model parameters are what we \emph{do not} want to model. If we could, they would already be part of the model, whose dynamics depend on other parameters (which, clearly, leads to a circular argument). An appropriate technique is to use robust control, in which model parameters are assumed to be uncertain, in a Knightian sense. This would allow us (at least theoretically) to set up and solve an optimization problem over multiple trading periods. The resulting optimal policies, however, tend to be very conservative, since they are mainly obtained with the ``worst-case scenario'' in mind.

\item We come up with a framework that is ``\emph{agnostic}'' to how model parameters evolve through the investment horizon. To be precise, we assume that the market parameters are stochastic processes within our assumed probability space, and that there is a model that describe them (i.e., they have a probability law). However, our objective is designed in a way that the optimal policies can be obtained \emph{without} the knowledge of the parameters law and just by observing their values once available.\footnote{In certain engineering and machine learning literature, such models are referred to as ``model-independent.'' Perhaps, a better terminology would be \emph{model-agnostic}.}
\end{itemize}

Predictable forward performance processes (PFPPs) fall into the last category. They are stochastic optimal control frameworks for an agent who controls a dynamically evolving system but can only prescribe the system dynamics for a short period ahead. PFPPs allow the agent to form time-consistent optimal policies over time horizons that spans multiple (possibly, infinite number of) model updates, thus
mitigating her model uncertainty.

Investments under PFPPs follow a forward-in-time iterative procedure. The agent starts by choosing a function $U_0(x)$, for all values of $x>0$. The function $U_0(\cdot)$ is interpreted as how she feels about (all possible values) of her initial wealth. She then estimates her first model. Assume that the estimated parameters are denoted by $\theta_1$, and the model is valid for the time period $[0,1]$. That is, the first trading period is $[0,1]$. She then solves an \emph{inverse} investment problem
\begin{align}\label{eq:InversInv1}
U_0(x) = \sup_X\Big\{\Eb[U_1(X,\theta_1)|\theta_1] : X\in \Ac_1(x,\theta_1)\Big\},\quad x>0,
\end{align}
in which $\Ac_1(x,\theta_1)$ denotes the set of all (random) admissible wealth at time 1 associated with wealth $x$ at time $0$ and assuming that the parameter values are $\theta_1$. Here, the utility function $U_1(\cdot,\theta_1)$ is unknown. In essence, \eqref{eq:InversInv1} is the inverse of a classical expected utility problem, in which the value function $U_0(\cdot)$ is known, and a terminal utility function $U_1(\cdot,\theta_1)$ is to be found. Once \eqref{eq:InversInv1} is solved, one also obtains an optimal wealth $X_1^{*,x}$ at time 1 and a corresponding trading strategy that replicates it over the first trading period (starting with an initial wealth $x$). At time 1, the model is updated and new values of the parameters, say $\theta_2$, are obtained over the second trading period. A second \emph{inverse} investment problem
\begin{align}\label{eq:InversInv2}
U_1(x,\theta_1) = \sup_X\Big\{\Eb[U_2(X,\theta_1,\theta_2)|\theta_1,\theta_2] : X\in \Ac_2(x,\theta_1,\theta_2)\Big\},\quad x>0,
\end{align}
is considered for the unknown function $U_2(\cdot,\theta_1,\theta_2)$. Solving this problem also determines the optimal policy for the second trading period $[1,2]$. This process is continued until the end of the investment horizon is reached. In particular, since the inverse investment problems are solved \emph{after} each parameter value is observed, the optimal policy can be executed \emph{without} the advance knowledge of the law of the parameters $\{\theta_n\}_{n=1}^{\infty}$.\footnote{Here, we are only concerned with pathwise construction of the optimal wealth process and the corresponding investment policy. That is, the optimal policy given a single path of the parameter values. We are \emph{not} claiming that we can obtain the unconditional law of the optimal wealth process without knowing the unconditional law of the parameters. Such a goal is, clearly, nonsensical.} The underlying idea is, thus, to choose an initial ``value function'' $U_0(\cdot)$, and then choose subsequent value functions $\{U_n(\cdot, \theta_1,\dots,\theta_n)\}_{n=1}^{\infty}$ by making sure that the next value function is consistent with the previous one after updating the model. By consistency, we mean that the functions $U_{n-1}(\cdot, \theta_1,\dots,\theta_{n-1})$ and $U_n(\cdot, \theta_1,\dots,\theta_n)$ satisfy a relationship similar to \eqref{eq:InversInv1} and \eqref{eq:InversInv2}.

The theory of PFPPs is still developing and the corresponding literature consists of only a handful of publications, mainly in the context of portfolio optimization in the binomial model. \cite{AZZ20} introduced PFPPs in the context of the binomial model, with one trade during each trading period. They established existence of PFPPs in this model, showed that their construction is done by solving a sequence of inverse investment problems, and solve the inverse investment problem (in a single period binomial model) by reducing it to a functional equation. \cite{Kallblad2020} studied, among other things, the single period inverse investment problem (which is used in construction of PFPPS) in the context of Black-Scholes model. Therein, it was shown that the inverse problem can be solved using the Weierstrass transform. \cite{StrubZhou2021} studied the (single period) inverse investment problem in the context of a complete semi-martingale market. They showed that this problem reduces to an integral equation which generalizes the functional equation in \cite{AZZ20}. \cite{LiangStrubWang2021} extended the results of \cite{AZZ20} to a binomial model with multiple trades in each trading period. Their theoretical results are close to the one in \cite{AZZ20}, in that construction of PFPPs reduces to a similar functional equation. They also studied a new application in robo-advising. In \cite{LiangStrubWang2023}, the same authors extended their model to multi-agents and mean-field game settings. \cite{Waldon2024} considered model ambiguity and analyzed robust PFPPs in the context of the binomial model and for linear and quadratic utility functions. To the best of our knowledge, the most general results on PFPPs so far is that of \cite{Angoshtari23}, who established existence of PFPP in the more general setting of conditionally complete markets (see Section \ref{sec:setup} below). Therein, it was established that existence and construction of PFPPs reduce to solving a sequence of integral equations, each of the form obtained in \cite{StrubZhou2021}. A new general approach for solving these integral equation based on Fourier transform was developed, as well as a closed-form solution under a special class of completely monotonic inverse marginal (CMIM) functions (see Section \ref{sec:CMIM} below).

It must be noted that the main idea of PFPPs originated from \emph{forward performance measurement}, which was proposed and extended in a series of papers by Musiela and Zariphopoulou, see \cite{MZ09, MZ10, MZ10-SPDE, MZ11}. The literature of forward performance measurement has since grown significantly. In the interest of keeping this introduction to a reasonable length, we will not provide a literature review and, instead, refer to \cite{LiangSunZariphopoulou2023} for a recent account of related work. There is a clear distinction between forward performance measurement and PFPPs. In particular, in the setting of forward performance measurement, the agent's preference is updated continuously in time and its dynamics can be modeled by differential equations, backward stochastic differential equations, or stochastic partial differential equations. In the framework of PFPPs, however, the agent's preference is updated in \emph{discrete time} and its dynamics are governed by inverse control problems and integral equations, which lead to fundamentally different mathematical theory than those developed for forward performance measurement.

Our main goal in this paper is to advance the theory of PFPPs by generalizing the results of \cite{Angoshtari23} beyond expected utility (EU). We choose rank-dependent utility (RDU) framework which was first proposed by \cite{Quiggin1982} and is also a key part of the Noble-prize-winning cumulative prospect theory of \cite{tversky1992advances}. RDU theory has been successful in explaining several empirically observed phenomena that the EU theory fails to explain. On the mathematical side, RDU framework poses a challenge since it yields an objective function that is both non-concave and time-inconsistent. These mathematical challenges were resolved through the quantile formulation introduced by \cite{JinZhou2008}, which lead to a significant advancement of RDU theory. In the interest of space, we refrain from providing an account of related work and, instead, refer the reader to one of the many excellent recent work such as \cite{HuJinZhou2021}.

Our main contribution is threefold. Firstly, we introduce the notion of \emph{rank-dependent predictable forward performance processes (RDPFPPs)} (see Definition \ref{def:RDPFPP}), which incorporates a random and exogenously evolving sequence of probability distortion functions into the framework of PFPPs. In other words, at the beginning of each trading period, the investor chooses (that is, estimates) both a model for the underlying market and a probability distortion for her objective in that period. Similar to the setting of PFPPs, we assume that the unconditional law for (the parameters of) the market models and the probability distortion functions are unknown. That is, we do not specify the dynamics of how the model and distortion functions evolve.

Secondly, we show that existence and construction of RDPFPPs reduces to solving an integral equation which generalizes the integral equation studied in \cite{StrubZhou2021} and \cite{Angoshtari23}, see Theorem \ref{thm:RDPFPP} and the integral equation \eqref{eq:IntEq}. In particular, we show that construction of RDPFPPs (and obtaining the corresponding optimal policy) does not require knowledge of the unconditional law of the model parameters and the probability distortion functions. Thus, this construction process is agnostic to the specific model driving the evolution of the distortion functions and market models.

Thirdly, we analyze the integral equation \eqref{eq:IntEq} by transforming it to a linear Volterra integral equation. To the best of our knowledge, we are the first in using the theory of Volterra integral equation in solving an inverse investment problem. Exploiting this connection allows us to obtain rather explicit formulae for the solution of the integral equation, see Propositions \ref{prop:Volterra_S}, \ref{prop:Volterra_RS}, \ref{prop:Volterra_Conc1}, and \ref{prop:Volterra_Conc2}. Our results cover cases when the special monotonicity assumption by \cite{JinZhou2008} holds, and when it does not hold. We also provide closed-form solutions for the case of completely-monotonic inverse marginal (CMIM) functions, see \eqref{eq:I0_CMIM_special} and Theorem \ref{thm:CMIM}. In particular, we prove that if the initial inverse marginal $I_0(\cdot)$ is CMIM, then so is the solution $I(\cdot)$ of the integral equation \eqref{eq:IntEq}.

Other minor contributions are as follows. We consider a conditionally complete market model with filtration that is more general than the one considered by \cite{Angoshtari23}. To illustrate our results, we use the conditionally complete Black-Scholes model in which trading is time-continuous. Note that existing examples (i.e. \cite{AZZ20}, \cite{LiangStrubWang2021}, and \cite{LiangStrubWang2023}) are mainly in the context of the simpler Binomial model for which trading is in discrete-time. Finally, our results for CMIM functions are proved for the general class of CMIM functions, while the results of \cite{Angoshtari23} were proved only for a special class of CMIM functions satisfying the representation \eqref{eq:CMIM_special}. See Remark \ref{rem:NotSpecial}.

We end this introduction by comparing our results to two related work. Firstly, the arguments in the proof of our first main result, namely, Theorem \ref{thm:RDPFPP}, are build on the approach developed by \cite{Xu2016} and, specifically, on the change-of-variable in Section 3 therein. Note, however, that we solve the inverse investment problem rather than the investment problem (that is, we assume that the value function is known and solve for the terminal utility function). See Remark \ref{rem:Xu2016} for further discussion.

The second related work is that of \cite{HeStrubZariphopoulou2021} which motivated our work. They were the first to merge RDU theory and forward performance measurement by introducing the notion of \emph{forward rank-dependent performance criteria} (FRDPC). An FRDPC is defined as a pair of deterministic processes $\big(\{u_t(\cdot)\}_{t\ge0},\{w_{s,t}(\cdot)\}_{0\le s<t}\big)$, in which $u_t(\cdot)$ is a utility function for agents' wealth at time $t$, and $w_{s,t}(\cdot)$ is a probability distortion function for time period $[s,t]$. In particular, since the agent's preference is updated in continuous time, there is a continuum of utility and distortion functions. This leads to a conflict between time-inconsistency of the RDU framework and time-consistency of forward performance criteria. To overcome this challenge, \cite{HeStrubZariphopoulou2021} started with two distinct definitions for FRDPC, one based on the counterpart of the martingale/supermartingale properties of forward performance criteria but under the probability distortion, and another based on time-consistency of candidate policies. They then proved that these two seemingly different notions are equivalent. This inherent conflict, however, still poses many challenges for further development of FRDPC. For instance, it is unclear how to introduce random utility and distortion functions. Furthermore, the definition of FRDPC significantly restricts the class of probability distortions for which FRDPC are not degenerate (i.e. they lead to non-zero allocation in the risky market). In particular, \cite{HeStrubZariphopoulou2021} shows that in the non-degenerate case, the probability distortion function must be of a specific form dictated by the quantile function of the pricing kernel (see equation (1) therein). They also prove that all viable distortion functions in FRDPC must satisfy the special monotonicity condition of \cite{JinZhou2008}.

In essence, RDPFPPs are the counterpart of FRDPC in which updating of preferences (and probability distortion functions) is done in discrete-time. Since RPPFPPs have no need for a continuum of utility and distortion functions, we do not face the conflict between time-inconsistency of the RDU framework and time-consistency of forward performance criteria. As such, our model is more flexible than that of \cite{HeStrubZariphopoulou2021} when it comes to the choice of utility and distortion functions. In particular, we can consider \emph{random} utility and distortion functions. Furthermore, our distortion functions do not necessarily need to satisfy \cite{JinZhou2008} monotonicity condition. On mathematical side, our arguments are fundamentally different from those in \cite{HeStrubZariphopoulou2021}, as we need to solve an inverse control problem and a corresponding integral equation.

The rest of the article is organized as follows. In Subsection \ref{subsec:notations}, for ease of reference, we provide notations that are commonly used throughout the paper. In Section \ref{sec:setup}, we define and motivate our market setup, namely, conditionally complete markets. Subsection \ref{subsec:GBM} provides the example of conditionally complete Black-Scholes market, to illustrate our (rather abstract) setup. In Section \ref{sec:RDPFPP}, we introduce RDPFPPs as a generalization of predictable forward performance processes that incorporate exogenous probability weighting functions. Section \ref{sec:Construction} includes our first main result, Theorem \ref{thm:RDPFPP}, which establishes existence of an RDPFPP through an iterative procedure whose main step is solving an integral equation. Section \ref{sec:IntEq} includes our second main result. It analyzes the integral equation and provides existence and uniqueness conditions for its solution, see Assumption \ref{asmp:Pho} and Propositions \ref{prop:Volterra_S}, \ref{prop:Volterra_RS}, \ref{prop:Volterra_Conc1}, and \ref{prop:Volterra_Conc2}. Section \ref{sec:CMIM} provides our third main result on regularity of the solution to the integral equation. In particular, Theorem \ref{thm:CMIM} states that if the initial inverse marginal is completely monotonic, then so is the solution. Section \ref{sec:Example_BS} provides numerical examples in the contexts of conditionally complete Black-Scholes market. We conclude in Section \ref{sec:conclude} and discuss some future research directions. Proofs are included in Appendices.

\subsection{Notations}\label{subsec:notations}
We use the following notations throughout the paper. $\Rb_+:=(0,+\infty)$. For $A\subseteq\Rb^m$, $\Cc^n(A)$ is the set of real-valued and continuously $n$-times differentiable functions with domain $A$, and $\Bc(A)$ is the set of all Borel subsets of $A$. Unless otherwise mentioned, equations, inequalities, and properties involving random variables are assumed to hold $\Pb$-almost surely with the probability space $(\Om,\Fc,\Pb)$ introduced in Section \ref{sec:setup}. The conditional probability of an event $E\in\Fc$ given a $\sig$-algebra $\Gc\subset\Fc$ is the $\Gc$-measurable random variable $\Pb(E|\Gc):=\Eb[\Ib_{E}|\Gc]$, in which $\Ib_{E}$ is the indicator of $E$. The cdf of a random variable $X$ is denoted by $F_X(x):=\Pb(X\le x)$, $x\in\Rb$, and its quantile function by $F^{-1}_X(q):=\sup\{x:F_X(x)\le q\}$, $0<q<1$. 
The conditional cdf of a random variable $X$ given a $\sig$-algebra $\Gc\subset\Fc$ is the $\Gc$-measurable random function $F_X(x|\Gc)(\om) := \Pb(X\le x|\Gc)(\om)$, $x\in\Rb$, and its conditional quantile function is the $\Gc$-measurable random function $F^{-1}_X(q|\Gc)(\om):= \sup\{x:F_X(x|\Gc)(\om)\le q\}$, $(q,\om)\in(0,1)\times\Om$.


%
%
\section{Conditionally complete market model}\label{sec:setup}

We consider an agent who manages a portfolio, consisting of a riskless asset and several risky assets, over multiple trading periods. For notational convenience, we take the first trading period to be the time period $[0,1)$, the second to be $[1,2)$, etc. At the beginning of each trading period, the agent makes two choices about that period, namely, a market model and a probability distortion function. For example, at time $n-1$, she may estimate the parameters of her market model and have full confidence in the estimated parameters over period $[n-1,n)$.\footnote{We leave uncertainty in model parameters as a future research direction. See the recent work \cite{Waldon2024}.} 
Similarly, she may consult with her clients periodically (say, as a robo-advisor would) and update her probability distortion function accordingly. Once a model and a probability distortion function are chosen at time $n-1$, the agent chooses an objective for the trading period $[n-1,n)$ and trade ``optimally'' in the market according to this objective. Our goal is to develop a systematic way of choosing these objectives. 
For the rest of this section, we specify our market model and the main assumptions we impose on it.

\begin{remark}\label{rem:noTrading}
To avoid  unnecessary clutter in the model setup and arguments, we will not follow the conventional approach of first introducing the price processes and then admissible trading policies. Instead, we assume existence of a conditional state-price density for each period (see Assumption \ref{asmp:rho}) and use them to define the set $\Ac$ of admissible wealth at the end of trading periods (see Definition \ref{def:Admis}). In other words, we assume that the agent only chooses a complete market at the beginning of each trading period (to be used for that period only). This assumption enables us to abstract away trading and focus on the set of admissible wealth at the end of the time period. Once an optimal end-of-period admissible wealth is specified, a corresponding trading strategy exists to replicate it, since the market is complete within the period.
This approach is justified since $\Ac$ and its properties are all that matter for characterizing and constructing RDPFPPs in the next section. 
At the end of this section, we will provide a concrete example to motivate our abstract setup and justify our main assumptions.
\qed
\end{remark}

Our model is defined within a complete probability space $(\Om, \Fc, \Pb)$ endowed with three filtrations $\Fb$, $\Gb$, and $\Fb^+$, satisfying the following assumption.
\begin{assumption}\label{asmp:Filtration}
There exist two complete filtrations $\{\Fc_n\}_{n=0}^\infty$ and $\{\Gc_n\}_{n=1}^\infty$ satisfying $\Gc_n \subset \Fc_n$ for $n\ge 1$. Furthermore, we define $\{\Fc_n^+\}_{n=0}^\infty$ in which $\Fc_n^+$ is the sigma-algebra generated by $\Fc_n\cup\Gc_{n+1}$.\qed
\end{assumption}

Note that $\Gc_{n+1}\subset \Fc^+_n$ and $\Fc_{n-1}^+\subset \Fc_n\subset \Fc_n^+$, while we have not assumed any direct relationship between $\Gc_{n+1}$ and $\Fc_n$.

In the context of the trading scenario discussed at the beginning of this section, these filtrations are interpreted as follows. 
We interpret $\Gc_n$, $n\ge1$, as  the $\sig$-field generated by the (random) parameters of the chosen market models and distortion functions of the first  $n$ periods, namely, trading periods $[0,1), \dots, [n-1,n)$. This implies that $\Gc_n$ is associated with time $n-1$ and just \emph{after} the $n$-th model and the distortion function is chosen. We interpret $\Fc_n$, $n\ge0$, as the information available at time $n$ and just \emph{before} the model and the distortion function for time period $[n,n+1)$ are estimated. In other words, $\Fc_n$ is the $\sig$-field generated by the price paths of assets over time period $[0,n]$ as well as (the parameters of) the chosen market models and distortion functions of the first $n$ periods $[0,1), \dots, [n-1,n)$. The condition $\Gc_n \subset \Fc_n$, $n\ge 1$, implies that at time $n$, and before choosing the $(n+1)$-th model and distortion function, we know the first $n$ models and distortion functions. Finally, $\Fc_n^+$ is interpreted as the information available at time $n$ and right \emph{after} the $(n+1)$-th model and distortion function are chosen.

We assume that the market model for each period $[n-1,n)$, $n\ge 1$, is conditionally complete given its parameter values. In particular, we assume that the set of end-of-period admissible wealth is characterized by the  \emph{conditional state-price density} $\{\rho_n\}_{n=1}^\infty$ satisfying the following conditions.

\begin{assumption}\label{asmp:rho}
For $n\ge1$, there exists a strictly positive, atomless, and $\Fc_n$-measurable random variable $\rho_n$ such that $\Eb\left[\rho_n\middle|\Fc^+_{n-1}\right]=1$ and that $\rho_n$ is conditionally independent of $\Fc_{n-1}^+$ given $\Gc_n$.\footnote{That is, $\Pb(\rho_n\in B|\Fc_{n-1}^+)=\Pb(\rho_n\in B|\Gc_n)$ for all Borel sets $B$.
}
We denote by $F_n(t|\Gc_n)$, $t>0$, the conditional cdf of $\rho_n$ given $\Gc_n$, and by $F_n^{-1}(q|\Gc_n)$, $0<q<1$, its conditional quantile function.\footnote{see subsection \ref{subsec:notations} for the definition of conditional cdf and quantile functions.}
\qed
\end{assumption}

As it will be discussed after Definition \ref{def:Admis} below, Assumption \ref{asmp:rho} implies a form of market completeness over time period $[n-1,n)$. The conditional independence assumption was introduced by \cite{Angoshtari23} (see Assumption 2.7 therein), and plays a central role in construction of RDPFPPs. It can be interpreted as the following statement: ``all the relevant parameters of the market model and distortion function for the time period $[n-1,n)$ are included in (generation of) $\Gc_n$.'' See \cite{Angoshtari23} for further discussion on the significance of this assumption. The assumption that $\rho_n$ is atomless is common in the rank-dependent utility literature\footnote{For instance, it is assumed in \cite{JinZhou2008}, \cite{CarlierDana2011}, \cite{HeZhou2011}, \cite{XiaZhou2016}, \cite{Xu2016}, \cite{HeKouwenbergZhou2017}, \cite{2019JinXiaZhou}, \cite{HeStrubZariphopoulou2021}, and \cite{HuJinZhou2021} that the pricing kernel is atomless.} and is adapted to avoid certain complications in the arguments.

\begin{remark}\label{rem:SubjectiveComplete}
It should be pointed out that we are assuming that the market is (conditionally) complete from the perspective of the agent. That is, at time $n-1$ and given $\Gc_n$, the agent believes that the market is complete and that the pricing kernel is $\rho_n$ over time period $[n-1,n)$. Our goal is to find an optimal control framework that accommodates such a belief and, in particular, allows for periodically updating the market model (and probability distortion function). For this reason, it is more appropriate to call our market model conditionally \emph{and subjectively (that is, subjective to the agent's belief)} complete.\qed
\end{remark}

Next, we define the set of admissible wealth processes observed at times 0,1,$\dots$ (i.e. beginning and end of trading periods). The definition relies on the conditional state-price densities $\{\rho_n\}_{n=1}^\infty$ in Assumption \ref{asmp:rho}.

\begin{definition}\label{def:Admis}
For any $n\ge 1$ and a non-negative $\Fc_{n-1}^+$-measurable random variable $\xi$, we define the set of all admissible (discounted) wealth at time $n$ starting with wealth $\xi$ at time $n-1$ as follows
\begin{align}\label{eq:An}
	\Ac_n(\xi) := \Big\{X\,:\, X\text{ is $\Fc_n$-measurable, } X\ge 0,\,
	\Eb\left[\rho_n X\middle|\Fc^+_{n-1}\right]= \xi\Big\}.
\end{align}
Furthermore,
\begin{align}\label{eq:Admiss}
	\Ac:=\Big\{\{X_n\}_{n=0}^\infty\,:\, X_n\text{ is $\Fc_n$-measurable, } X_n\ge 0,\, \Eb[\rho_n X_n|\Fc^+_{n-1}]= X_{n-1} \Big\},
\end{align}
is the set of all admissible wealth processes.  That is, we define an admissible (discounted) wealth process as an $\Fb$-adapted process $X=\{X_n\}_{n=0}^\infty$ such that $X_0\ge0$ and $X_n\in\Ac_n(X_{n-1})$ for all $n\ge1$.\qed
\end{definition}

In Definition \ref{def:Admis}, there is subtlety in choosing the filtration with respect to which $X_n$ is measurable. Since $X_n$ is the outcome of trading over time period $[n-1, n)$, it cannot depend on the parameter values and the distortion function of the next time period $[n,n+1)$. 
Thus, $X_n$ must be $\Fc_{n}$-measurable and not just $\Fc_{n}^+$-measurable. 

Note also that the definition of $\Ac_n$ in \eqref{eq:An} implies a form of market completeness over time period $[n-1,n)$. In fact, any non-negative claim $X_n$ at time $n$ is assumed to be replicable starting with an $\Fc_{n-1}^+$-measurable wealth $\xi$ at time $n-1$ if: $(1)$ $X_n$ is $\Fc_n$-measurable (and not just $\Fc_n^+$-measurable), and $(2)$ $\Eb[\rho_n X_n|\Fc^+_{n-1}]= \xi$. By the first fundamental theorem of asset pricing, this implies that the underlying market model for the time period $[n-1,n)$ is complete conditionally on $\Fc_{n-1}^+$, with $\rho_n$ as its unique (conditional) state-price density.

Let us also emphasize another fact: our market model is allowed to be incomplete over multiple periods. For instance, not every claim at time $2$ can be replicated over time period $[0, 2)$. The only claims $X_2$ at time $2$ that can be replicated over $[0,2)$ are those that are $(1)$ $\Fc_2$-measurable (and not just $\Fc_2^+$-measurable), and $(2)$ satisfy $\Eb[\rho_2 X_2|\Fc_1^{+}]=X_1$ for some $\Fc_1$ (and not just $\Fc_1^+$-) measurable claim $X_1$. See \cite{Angoshtari23} for further discussion on conditional completeness.

\begin{remark}\label{rem:EMM}
We have assumed that the conditional state-price densities $\{\rho_n\}_{n=1}^\infty$ are unique. This assumption does not contradict with possible incompleteness of the market model over multiple periods, because the unconditional state-price densities may not be unique. To highlight this, we characterize the unconditional state-price densities in our setting. 
In terms of the unconditional state-price densities, the set of admissible wealth should be of the form
\begin{align}
	\Ac':=\Big\{\{X_n\}_{n=0}^\infty\,:\, &X_n\ge0, X_n\text{ is $\Fc_n$-measurable, }\\
	&\Eb[Z_n X_n|\Fc^+_{n-1}]\le Z_{n-1}X_{n-1}
	\text{ for all } \{Z_n\}_{n=0}^\infty\in\Mc\Big\},
\end{align}
in which $\Mc$ is the set all (unconditional) state-price densities, that is, the set of all $\{\Fc^+_n\}_{n=0}^\infty$-adapted local martingale deflators of the discounted prices of the underlying assets. By comparing $\Ac'$ with \eqref{eq:Admiss}, we conclude that any $\{\Fc^+_n\}_{n=0}^\infty$-adapted positive process $\{Z_n\}_{n=0}^{\infty}$ satisfying
\begin{align}\label{eq:Z_rho}
	\Eb\left[\frac{Z_n}{Z_{n-1}}\middle|\Fc_n\right]=\rho_n, \quad n\ge1,
\end{align}
is an unconditional state-price density.	In particular, such unconditional state-price densities $\{Z_n\}_{n=0}^\infty$ can be non-unique even if the process $\{\rho_n\}_{n=1}^\infty$ is unique.  Thus, uniqueness of the conditional state-price density $\rho_n$ does not imply that the market is complete. Note also that the unconditional price densities $\{Z_n\}_{n=0}^{\infty}$ play no role in our future arguments.\qed
\end{remark}

\subsection{Conditionally complete Black-Scholes market}\label{subsec:GBM}
We end this section by providing an example of our market setting. 
Consider a $d-$dimensional standard Brownian motion $B=(B_t)_{t\ge 0}$, a sequence of $\Rb^d$-valued random vectors $\{\lambda_n\}_{n=1}^\infty$, and a sequence of $\Rb^{d\times d}$-valued non-singular random matrices $\{\sigma_n\}_{n=1}^\infty$ in a probability space $(\Om,\Fc,\Pb)$. Assume further that the increment $(B_t-B_{n-1})$ is independent of $\{(\lambda_m, \sigma_m)\}_{m=1}^{m=n}$ for all $n\in\{1,\dots\}$ and all $t>n-1$. Denote $\Fc_t$ as the augmented $\sigma$-field generated by $\left\{(B_s)_{0\le s\le t}, \{(\lambda_m,\sig_m)\}_{m=1}^{\lceil t \rceil}\right\}$, and let $\Gc_n$ be a $\sig$-field containing the augmented $\sig$-field generated by $\{(\lam_m,\sig_m)\}_{m=1}^n$. 
According to Assumption \ref{asmp:rho}, we define $\Fc_n^+$ as the augmented $\sig$-field generated by $\big\{(B_s)_{0\le s\le n},$ $\{(\lambda_m,\sig_m)\}_{m=1}^{n+1}\big\}$.

Note that $(B_s)_{0\le s\le t}$ is $\Fc_t$-measurable, while $(\lambda_n, \sigma_n)$ is $\Fc_{n-1}^+$-measurable and $\Fc_t$-measurable for $t>n-1$. Note also that $\Fc_{n^+}=\bigcap_{t>n} \Fc_t = \Fc^+_n$.

Now, consider a market with one riskless asset and $d$ stocks. The interest rate is assumed to be zero, that is, we take the riskless asset as our numeraire. The (discounted) price of the risky assets, denoted by $S_t=(S^1_t,\dots,S_t^d)$, $t\ge 0$, is given by $S_0 = s_0\in(0,+\infty)^d$ and
\begin{align}\label{eq:GBM_S}
	S^i_t = S^i_{n-1}\exp\left[\left(\sig^{i\cdot}_n\cdot\lam_n-\frac{1}{2}\|\sig^{i\cdot}_n\|^2\right)(t-n+1) + \sig_n^{i\cdot}\cdot\left(B_t - B_{n-1}\right)\right],
\end{align}
for $i\in\{1,\dots,d\}$, $n\in\{1,\dots\}$, and $t\in(n-1,n]$. Here, $\sig_n^{i\cdot}$ denotes the $i$-th row of $\sig_n$. Since $\sig_n$ is non-singular, we can express $(B_t-B_{n-1})$ in terms of $\log(S_t/S_{n-1})$. Thus, the (augmented) $\sigma$-field generated by $\left\{(S_u)_{0\le u\le t}, \{(\lambda_m,\sig_m)\}_{m=1}^{\lceil t \rceil}\right\}$ coincides with $\Fc_t$.

In the filtered space $(\Om, \Fc, \Pb, \{\Fc_t\}_{n-1<t\le n})$ and with the Brownian motion $(B_t)_{n-1\le t\le n}$, the process $S_t=(S^1_t,\dots,S_t^d)$, $t\in(n-1,n]$, is the unique strong solution of 
\begin{align}
	S_t = S_{n-1} + \int_{n-1}^t \diag(S_u)\sigma_n(\lambda_n \dd u + \dd B_u), \quad n-1< t\le n.
\end{align}
We thus interpret $\sig_n\sig_n^\top$ (respectively, $\lam_n$) as the covariance matrix (respectively, vector of risk-premiums) of the stocks over time period $(n-1,n]$. Note that, unlike the Black-Scholes market model, the covariance matrix and risk-premiums are random. Furthermore, since $(\sig_n,\lam_n)$ are $\Fc_{(n-1)^+}$-measurable, we have assumed that the covariance matrix and risk-premiums are observed (that is, estimated with full confidence) at the beginning of the $n$-th period.

A trading strategy $(\pi_t)_{t\ge0}$ is an $(\Fc_t)_{t\ge0}$-adapted process satisfying $\Pb\left(\int_0^T\pi^2_t\dd t<\infty\right)=1$ for all $T>0$. The wealth process $(X_t)_{t\ge0}$ associated with a trading strategy $(\pi_t)_{t\ge0}$ and starting with initial wealth $x>0$ is the solution of equations
\begin{align}
	X_t = X_{n-1} + \int_{n-1}^t \pi_u \dd S_u,\quad n-1< t\le n, n\in\{1,\dots\},
\end{align}
with $X_0=x$. A trading strategy is admissible if $X_t\ge 0$ for all $t\ge0$. It is important to observe that $X_n$ is $\Fc_n$-measurable and not just $\Fc_{n}^+$-measurable. In particular, $X_n$ does not depend on the covariance matrix $\sig_{n+1}$ and the risk premiums $\lam_{n+1}$ which are also observed at time $n$. This is expected (or, rather, it is a required property for the wealth process), since the mere act of estimating a model should not instantaneously change wealth.

From \eqref{eq:GBM_S}, it follows that an unconditional state price density (i.e. a local martingale deflator for $\{S_t\}_{t\ge0}$) is a process $(Y_t)_t\ge0$ of the form
\begin{align}
	Y_t = Z_{n-1} \exp\left[-\frac{1}{2}\|\lam_n\|^2(t-n+1) - \lam_n\cdot(B_t-B_{n-1})\right],\quad t\in[n-1, n), n\in\{1,\dots\},
\end{align}
in which $\{Z_n\}_{n=0}^{\infty}$ is an $\{\Fc^+_n\}_{n=0}^\infty$-adapted positive process satisfying $Z_0=1$ and
\begin{align}\label{eq:GBM_Z}
	\Eb[Z_n|\Fc_n] = Z_{n-1} \exp\left[-\frac{1}{2}\|\lam_n\|^2 - \lam_n\cdot(B_n-B_{n-1})\right],\quad n\in\{1,\dots\}.
\end{align}
The last equation only specifies the conditional expectation of the $\Fc_n^+$-measurable random variable $Z_n$. Thus, the (unconditional) state-price densities are not unique. By \eqref{eq:Z_rho}, the conditional state-price densities $\{\rho_n\}_{n=1}^\infty$ are given by
\begin{align}\label{eq:GBM_rho}
	\rho_n := \exp\left[-\frac{1}{2}\|\lam_n\|^2 - \lam_n\cdot(B_n-B_{n-1})\right],\quad n\in\{1,\dots\}.
\end{align}
Thus, $\rho_n$ is conditionally log-normal given $\Gc_n$. Its conditional cdf is given by
\begin{align}
	F_n(t|\Gc_n) = \frac{1}{2}+\frac{1}{2}\erf\left(\frac{\ln(t)}{\sqrt{2}\|\lam_n\|} +\frac{1}{2\sqrt{2}}\|\lam_n\|\right),\quad t>0,
\end{align}
and its quantile functions is given by
\begin{align}\label{eq:LogNormQuant}
	F_n^{-1}(q|\Gc_n) = \exp\left(\sqrt{2}\|\lam_n\|\erf^{\,-1}(2q-1)-\frac{1}{2}\|\lam_n\|^2\right),\quad 0<q<1,
\end{align}
in which $\erf(z)=(2/\sqrt{\pi})\int_0^z \ee^{-x^2}\dd x$ is the error function. It is straightforward to check that Assumption \ref{asmp:rho} is satisfied.

%
%
\section{Definition of RDPFPPs}\label{sec:RDPFPP}

In this section, we introduce RDPFPPs as a generalization of predictable forward performance processes introduced by \cite{AZZ20}. 
The new element is the addition of an exogenous probability weighting function $W_n(x)$ for each trading period $[n-1,n)$. These so-called ``distortion functions'' determine how the agent subjectively perceives the possibility of rare events in each period. We assume that the probability weighting functions are random and are revealed at the beginning of each period (i.e.  $W_n(\cdot)$ is $\Gc_n$-measurable). Specifically, we assume that:

\begin{assumption}\label{asmp:ProbDistortion}
	For $n\in\{1,\dots\}$, there exists a random probability weighting function
	$W_n(p,\om):[0,1]\times\Om\to[0,1]$ satisfying the following conditions:\\[1ex]
	\noindent$(i)$ $(p,\om)\mapsto W_n(p,\om)$ is $\Bc([0,1])\times\Gc_n$-measurable.\\
	\noindent$(ii)$ For $\Pb$-almost all $\om$, the mapping $W_n(\cdot,\om):[0,1]\to [0,1]$ is continuously differentiable and strictly increasing with $W_n(0,\om)=0$ and $W_n(1,\om)=1$.\qed
\end{assumption}

We define an RDPFPP next. In the definition, $F_X(x|\Gc)(\om) := \Pb(X\le x|\Gc)(\om)$ denotes the conditional cdf of a random variable $X$ given a $\sig$-algebra $\Gc\subset\Fc$, and the set
\begin{align}\label{eq:Uc}
	\Uc:=\Big\{U\in \Cc^2(\Rb_+): U'>0, U''<0, U'(0)=+\infty, U'(\infty)=0\Big\},
\end{align}
is the set of all classical utility functions.

\begin{definition}\label{def:RDPFPP}
	A sequence of random functions $\big\{U_n(\cdot)=U_n(\cdot,\om)\big\}_{n=0}^\infty$ is an RDPFPP if the following conditions are satisfied.\\[1ex]
	\noindent$(i)$ $U_0\in \Uc$. For $n\ge1$, $U_n(x,\om)$ is $\Bc(\Rb_+)\times\Gc_n$-measurable and $U_n(\cdot,\om)\in\Uc$ for $\Pb$-almost all $\om$.\\
	\noindent$(ii)$ $U_{n-1}(X_{n-1}) \ge \int_0^\infty U_n(\xi)\, \dd \big(1-W_n(1-F_{X_n}(\xi|\Gc_n)\big)$, $n\ge1$, for any $\{X_n\}_{n=0}^\infty\in\Ac$.\\
	\noindent$(iii)$ There exists $\{X^*_n\}_{n=0}^\infty\in\Ac$ such that 
	$U_{n-1}(X^*_{n-1}) = \int_0^\infty U_n(\xi)\, \dd \big(1-W_n(1-F_{X^*_n}(\xi|\Gc_n)\big)$, $n\ge1$. Such an $X^*$ is called an \emph{optimal wealth process}.\qed
\end{definition}

An RDPFPP is a sequence of random utility functions $\{U_n(\cdot)\}$ for the agent's wealth at the beginning and the end of trading periods. By Condition $(i)$ of Definition \ref{def:RDPFPP}, $U_0(x)$ is known at time 0 and before the model and distortion function for the first period are observed. For $n\ge1$, $U_n(\cdot)$ is known at time $n-1$ and after the market model and distortion function for the $n$-th period (i.e. $[n-1,n)$) are revealed, since $U_n$ is $\Gc_n$-measurable but not necessarily $\Fc_{n-1}$-measurable.

The role of Conditions $(ii)$ and $(iii)$ is twofold. Firstly, they enforce a rank-dependent preference within each period. 
Consider an optimal investment problem with a rank dependent preference over a single time horizon, namely,
\begin{align}\label{eq:RDU}
	v(x) := \sup_{X\in\Ac(x)} \int_0^{+\infty} u(\xi)\,\dd\big[1-w\big(1-F_X(\xi)\big)\big],
\end{align} 
in which $\Ac(x)$ is the set of admissible terminal wealth starting with initial wealth $x$, $u(\cdot)$ is the terminal utility function, $w(\cdot)$ is a probability weighting function, and $v(\cdot)$ is the value function. Conditions $(ii)$ and $(iii)$ of Definition \ref{def:RDPFPP} imply that, for all $n\in\{1,\dots\}$, all $X\in\Ac_n\big(X^*_{n-1}\big)$ (see \eqref{eq:An}), and $\Pb$-almost all $\om$, we have
\begin{align}
	U_{n-1}\big(X_{n-1}^*(\om), \om\big)
	&= \int_0^{+\infty} U_n(x,\om) \dd \left[1-W_n\big(1-F_{X^*_n}(x|\Gc_n)(\om),\om\big)\right]\\*
	\label{eq:Optimality}
	&\ge \int_0^{+\infty} U_n(x,\om) \dd \left[1-W_n\big(1-F_{X}(x|\Gc_n)(\om),\om\big)\right].
\end{align}
It is then clear that \eqref{eq:Optimality} is the conditional version of the rank-dependent preference, with $U_{n-1}$, $U_n$, and $W_n$ respectively playing the role of the value function $v$, the utility function $u$, and the probability weighting function $w$, in \eqref{eq:RDU}.

The second role of Conditions $(ii)$ and $(iii)$ is to guarantee time-consistency of the agent over multiple periods. Assume that the agent starts with an initial wealth $x>0$ and adapts an RDPFPP $\{U_n\}_{n=1}^{+\infty}$ in the following sense. At time $n-1$ and once the market model and probability weighting function for $[n-1,n)$ are known (i.e. the information in $\Gc_n$ is revealed), the agent chooses an admissible wealth $X_n\in\Ac_n(X_{n-1}^*)$ to maximize her single period objective
\begin{align}
	\int_0^{+\infty} U_n(x,\om) \dd \big[1-W_n\big(1-F_{X_n}(x|\Gc_n)(\om),\om\big)\big].
\end{align}
By \eqref{eq:Optimality}, it follows that the agent would choose $X^*_n$ as her end-of-period terminal wealth. The important observation here is that \eqref{eq:Optimality} is ``pathwise,'' that is, it holds for ($\Pb$-almost) all $\om$. Therefore, by adapting the RDPFPP $\{U_n\}_{n=1}^{+\infty}$, the agent is essentially enforcing \eqref{eq:Optimality} for all future times and for all future realization of market models and probability distortion functions. This indicates that the agent preference is strongly time-consistent in the sense of \cite{Riedel2004}.

It remains to show that such an RDPFPP exists and, more importantly, can be constructed within the context of the investment scenario specified in Section \ref{sec:setup}. This will be our main goal for the rest of the paper.

%
%
\section{Forward construction of RDPFPP}\label{sec:Construction}
This section includes one of our main results, namely, Theorem \ref{thm:RDPFPP} below. It establishes existence of an RDPFPP $\{U_n\}_{n=0}^\infty$ through an iterative procedure in which the random utility function $U_n(\cdot,\om)$ is obtained from $U_{n-1}(\cdot,\om)$. Each iteration of this procedure is akin to solving an \emph{inverse} investment problem under a rank-dependent preference, in the sense that the value function $v(\cdot)$ in \eqref{eq:RDU} is known and the terminal utility function $u(\cdot)$ therein is to be found.

In the statement of Theorem \ref{thm:RDPFPP}, $\Ic$ is the set of all inverse marginal of utility functions in \eqref{eq:Uc}, that is,
\begin{align}\label{eq:Ic}
	\Ic := \left\{\left(U'\right)^{-1}:U\in \Uc\right\} = \Big\{I\in \Cc^1(\Rb_+): I'<0, I(0)=+\infty, I(+\infty)=0\Big\}.
\end{align}
We also need the random functions $\Phi_n(\cdot,\om)$, $n\ge1$, given by
\begin{align}
	\label{eq:Phi}
	\Phi_n(q,\om) &:= -\int_q^1 F^{-1}_n\big(W_n^{-1}(1-p,\om)\big|\Gc_n\big)(\om)W_n^{(-1)\prime}(1-p,\om)\dd p
	\\&
	=-\int_0^{W_n^{-1}(1-q,\om)} F^{-1}_n\big(p\big|\Gc_n\big)(\om)\dd p,
\end{align}
for $(q,\om)\in[0,1]\times\Om$, in which $F_n^{-1}(\cdot|\Gc_n)$ is the conditional quantile function of $\rho_n$ in Assumption \ref{asmp:rho}, and $W_n(\cdot,\om)$ is the random probability weighting function in Assumption \ref{asmp:ProbDistortion}. Note that, for $\Pb$-almost all $\om$, $\Phi_n(\cdot,\om)$ is absolutely continuous, strictly increasing, $\Phi_n(1,\om)=0$, and $\Phi_n(0,\om)=-\Eb[\rho_n|\Gc_n](\om)=-1$. We denote by $\Pho_n(\cdot,\om)$ the concave envelope of $\Phi_n(\cdot,\om)$, that is
\begin{align}\label{eq:Pho}
	\Pho_n(q,\om):=\sup_{\al,\beta}\left\{\frac{(\beta-q)\Phi_n(\alpha,\om)+(q-\alpha)\Phi_n(\beta,\om)}{\beta-\alpha}\,:\,
	\al\in[0,q], \beta\in[q,1]\right\},\quad q\in[0,1].\qquad
\end{align}
Since $\Phi_n(\cdot,\om)$ is absolutely continuous, so is $\Pho_n(\cdot,\om)$. Let $\Phi'_n(\cdot,\om)$ and $\Pho'_n(\cdot,\om)$ denote the (almost everywhere) derivatives of $\Phi_n(\cdot,\om)$ and $\Pho_n(\cdot,\om)$, respectively. Note that
\begin{align}\label{eq:Phi_prime}
	\Phi_n'(q,\om) = F^{-1}_n\big(W_n^{-1}(1-q,\om)\big|\Gc_n\big)(\om)W_n^{(-1)\prime}(1-q,\om)
	=\frac{F^{-1}_n\big(W_n^{-1}(1-q,\om)\big|\Gc_n\big)(\om)}{W_n'\big(W_n^{-1}(1-q,\om),\om\big)}.
\end{align}

We now state our main result on existence of RDPFPP. This result generalizes Theorem 3.4 of \cite{Angoshtari23}.
\begin{theorem}\label{thm:RDPFPP}
	Consider the market model of Section \ref{sec:setup} with Assumptions \ref{asmp:Filtration}, \ref{asmp:rho}, and \ref{asmp:ProbDistortion} holding, and let $\{\Pho_n(\cdot,\om)\}_{n=1}^\infty$ be as in \eqref{eq:Pho}. Assume further that $U_0\in\Uc$, $I_0:=\left(U_0'\right)^{-1}$, and that the random functions $\big\{I_n(\cdot,\om)\big\}_{n=1}^\infty$ satisfy the following conditions:\\[1ex]
	\noindent$(i)$ For $n\ge1$, $I_n(y,\om)$ is $\Bc(\Rb_+)\times\Gc_n$-measurable and $I_n(\cdot,\om)\in\Ic$ for $\Pb$-almost all $\om$.\\
	\noindent$(ii)$ $\int_0^1 I_n\big(y\Pho_n'(\eta,\om),\om\big)\Pho_n'(\eta,\om)\dd \eta=I_{n-1}(y,\om)$, for all $y>0$ and $\Pb$-almost all $\om$.\\
	\noindent$(iii)$ $\int_0^1\int_{I_n\big(\Pho_n'(p,\om),\om\big)}^x I_n^{-1}(\xi,\om)\dd \xi\dd p<+\infty$, for all $x>0$ and $\Pb$-almost all $\om$.\\[1ex]
	Then, the sequence of random functions $\{U_n(x,\om)\}_{n=0}^{\infty}$ given by the recursive relationship
	\begin{align}\label{eq:RDPFPP}
		U_n(x,\om)
		= U_{n-1}\big(I_{n-1}(1,\om),\om\big)+ \int_0^1\int_{I_n\big(\Pho_n'(p,\om),\om\big)}^x I_n^{-1}(\xi,\om)\dd \xi\dd p,
		\quad x>0, \om\in\Om, n\ge 1,\qquad
	\end{align}
	is an RDPFPP. Furthermore, for any initial portfolio value $x_0>0$, the process $\{X^*_n\}_{n=0}^\infty$ given by $X^*_0=x_0$ and the recursive relationship
	\begin{align}\label{eq:RDPFPP_optWealth}
		X^*_n = I_n\Bigg(U_{n-1}'\left(X^*_{n-1}\right)\Pho'_n\bigg(1-W_n\Big(F_n\big(\rho_n\big|\Gc_n\big)\Big)\bigg)\Bigg),\quad n\ge1,\quad
	\end{align}
	is an optimal wealth process.\qed\vspace{1ex}
\end{theorem}

\begin{proof}
	See Appendix \ref{app:RDPFPP}.
\end{proof}

\begin{remark}\label{rem:Xu2016}
	The proof of Theorem \ref{thm:RDPFPP} is based on the arguments in \cite{Xu2016} (which are referenced as needed). In particular, we use the change-of-variable in Section 3 of \cite{Xu2016}. The difference between our setting and that of \cite{Xu2016} (that is, our mathematical contribution) is that we are solving an \emph{inverse} investment problem. That is, we are assuming that value function $v(\cdot)$ in \eqref{eq:RDU} is known and the terminal utility function $u(\cdot)$ therein is to be found. See \eqref{eq:RDProblem} in Appendix \ref{app:RDPFPP} for details.
	\qed
\end{remark}

Theorem \ref{thm:RDPFPP} implies the following iterative construction for an RDPFPP $\{U_n(\cdot,\om)\}_{n=0}^{+\infty}$ along with a corresponding optimal wealth process $\{X^*_n\}_{n=0}^\infty$. Assume that an initial utility function $U_0(x)$, for all initial wealth $x$, is chosen in advance (that is, at time zero and before the model and probability weighting function for the first period $[0,1)$ are observed). Assume that the initial portfolio value is $x_0>0$. Let $I_0=(U'_0)^{-1}$ be the inverse marginal of $U_0$. At time $n-1$ and assuming that $U_{n-1}(\cdot,\om)$, $I_{n-1}(\cdot,\om)$, $\Pho_n(\cdot,\om)$, and $X_{n-1}^*(\om)$ are known, follow these steps:
\begin{description}
	\item[\textbf{Step 1:}] Obtain $I_n(\cdot,\om)$ by solving the integral equation
	\begin{align}\label{eq:IntEq_n}
		\int_0^1 I_n\big(y\Pho_n'(\eta,\om),\om\big)\Pho_n'(\eta,\om)\dd \eta=I_{n-1}(y,\om),\quad y>0. 
	\end{align}
	\item[\textbf{Step 2:}] Obtain $U_n(\cdot,\om)$ from \eqref{eq:RDPFPP}. Note that $\big(U_n'(\cdot,\om)\big)^{-1}=I_n(\cdot,\om)$.
	\item[\textbf{Step 3:}] Obtain $X_n^*$ from \eqref{eq:RDPFPP_optWealth}. It is shown in the proof of Theorem \ref{thm:RDPFPP} below that $X^*_n\in\Ac_n(X^*_{n-1})$. Thus, there exists a trading strategy that starts with wealth $X^*_{n-1}$ at time $n-1$ and yields wealth $X^*_n$ at time $n$. Follow this trading strategy over time period $[n-1,n)$.
\end{description}
Repeating Steps 1-3 yields an RDPFPP $\{U_n(\cdot,\om)\}_{n=0}^{+\infty}$ along with a corresponding optimal wealth process $\{X^*_n\}_{n=0}^\infty$. This construction is forward in time, in that we only need past information to obtain $U_n$ and $X^*_n$. Note also that the construction is \emph{pathwise}. That is, to obtain a specific ``path'' $\om$, the realization $U_n(\cdot,\om)$ is obtained from past information for the same path $\om$.

Theorem \ref{thm:RDPFPP} thus reduces construction of RDPFPPs (and a corresponding optimal wealth process) to solving the integral equation \eqref{eq:IntEq_n}. The solution of this integral equation will be discussed in Sections \ref{sec:IntEq} and \ref{sec:CMIM}.

The integral equation \eqref{eq:IntEq_n} generalizes the integral equation in \cite{StrubZhou2021} and \cite{Angoshtari23}, which is obtained under the expected utility framework. To see this, assume that the probability weighting functions are identity functions, i.e. $W_n(p,\om)=p$ for $p\in[0,1]$. We then have
\begin{align}
	\int_0^\infty U_n(\xi)\, \dd \big(1-W_n(1-F_{X_n}(\xi|\Gc_n)\big)
	=\int_0^\infty U_n(\xi)\, \dd F_{X_n}(\xi|\Gc_n)
	=\Eb\left[U_n(X_n)|\Gc_n\right],
\end{align}
and an RDPFPP $\{U_n(\cdot,\om)\}_{n=0}^{+\infty}$ in Definition \ref{def:RDPFPP} becomes a PFPP (see Definition 3.1 in \cite{Angoshtari23}). From \eqref{eq:Phi}, we obtain
\begin{align}
	\Phi_n'(p,\om) &:= F^{-1}_n\big(W_n^{-1}(1-p,\om)\big|\Gc_n\big)(\om)W_n^{(-1)\prime}(1-p,\om)
	=F^{-1}_n\big(1-p\big|\Gc_n\big)(\om).
\end{align}
Thus, $\Phi_n(\cdot,\om)$ is concave and $\Pho_n'(p,\om)=\Phi_n'(p,\om)=F^{-1}_n\big(1-p\big|\Gc_n\big)(\om)$, $p\in[0,1]$. It then follows that
\begin{align}
	&\int_0^1 I_n\big(y\Pho_n'(\eta,\om),\om\big)\Pho_n'(\eta,\om)\dd \eta
	=\int_0^1 I_n\big(y F_n^{-1}(1-\eta|\Gc_n)(\om),\om\big)F_n^{-1}(1-\eta|\Gc_n)(\om)\dd \eta\\
	&=\int_0^{+\infty} I_n(y \xi,\om)\xi\, \dd\big(F_n(\xi|\Gc_n)(\om)\big)
	=\Eb\left[\rho_n I_n(y\rho_n)|\Gc_n\right](\om).
\end{align}
Condition $(ii)$ of Theorem \ref{thm:RDPFPP} then reduces to Condition $(ii)$ of Theorem 3.4 of \cite{Angoshtari23}, and the integral equation \eqref{eq:IntEq_n} becomes
\begin{align}
	\int_0^{+\infty} I_n(y \xi,\om)\xi\, \dd\big(F_n(\xi|\Gc_n)(\om)\big) = I_{n-1}(y,\om),\quad y>0.
\end{align}
This is the integral equation in  \cite{StrubZhou2021} and \cite{Angoshtari23}.

%
%
\section{The integral equation: existence and uniqueness}\label{sec:IntEq}
In the previous section, we showed that solving the pathwise integral equation \eqref{eq:IntEq_n} is the main step in the forward-in-time construction of RDPFPPs. By ignoring dependence on $n$ and $\om$, this integral equation becomes
\begin{align}\label{eq:IntEq}
	\int_0^1 I\big(y\Pho'(\eta)\big)\Pho'(\eta)\dd\eta = I_0(y),\quad y>0,
\end{align}
in which $I_0(\cdot)$ and $\Pho(\cdot)$ are known, and $I(\cdot)$ is unknown. In this section, we provide conditions on $I_0(\cdot)$ and $\Pho(\cdot)$ that guarantee existence and uniqueness of a continuous solution $I(\cdot)$ to \eqref{eq:IntEq}. We postpone discussion on the regularity of the solution to the next section. That is, in Section \ref{sec:CMIM}, we will provide additional conditions for the solution $I(\cdot)$ to be an inverse marginal (i.e. $I(\cdot)\in\Ic$).

We will provide the solution of \eqref{eq:IntEq} in Propositions \ref{prop:Volterra_S}, \ref{prop:Volterra_RS}, \ref{prop:Volterra_Conc1}, and \ref{prop:Volterra_Conc2} below and under different conditions. Our main approach (and the unifying theme of this section) is to first transform \eqref{eq:IntEq} into a linear Volterra integral equation of the form
\begin{align}\label{eq:Volterra}
	J(t) = J_0(t) + \int_0^t J(s) k(t,s)\dd s,\quad t>0,
\end{align}
in which $J(\cdot)$ is unknown and $J_0(\cdot)$ and $k(\cdot,\cdot)$ are known. It is well-known that, under certain conditions, the unique solution of \eqref{eq:Volterra} is  given by
\begin{align}\label{eq:Volterra_sol}
	J(t) = J_0(t) + \int_0^t J_0(s) \ks(t,s)\dd s,\quad t>0,
\end{align}
in which $\ks(\cdot,\cdot)$ is the so-called \emph{resolvent kernel} associated with $k(\cdot,\cdot)$. Our solutions for the integral equation \eqref{eq:IntEq} (namely, \eqref{eq:IntEq_sol_S}, \eqref{eq:IntEq_sol_RS}, \eqref{eq:IntEq_sol_Conc1}, and \eqref{eq:IntEq_sol_Conc2}) are all in the form of \eqref{eq:Volterra_sol}. In particular, the resolvent kernel $\ks(\cdot,\cdot)$ only depends on the function $\Pho(\cdot)$, while the ``initial data'' $J_0(\cdot)$ depends both on $\Pho(\cdot)$ and $I_0(\cdot)$. Therefore, one can calculate the resolvent kernel $\ks(\cdot,\cdot)$ once, and then use the formula \eqref{eq:Volterra} to solve \eqref{eq:IntEq} for different $I_0(\cdot)$.

Assumption \ref{asmp:Pho} below is our main existence-uniqueness condition for the solution of the integral equation \eqref{eq:IntEq}. Note that $\Pho(\cdot)$ in \eqref{eq:IntEq} represents the concave envelope of the function $\Phi_n(\cdot,\om)$ in \eqref{eq:Phi}. Thus, Assumption \ref{asmp:Pho} can be seen as imposing further restrictions on conditional distributions of $\{\rho_n\}_{n=1}^\infty$ in Assumption \ref{asmp:rho}, as well as the probability weighting functions $\{W_n\}_{n=1}^\infty$ of Assumption \ref{asmp:ProbDistortion}. 

\begin{assumption}\label{asmp:Pho}
	$\Pho(\cdot)$
	satisfies one of the assumptions (a), (b), or (c) below:\\[1ex]
	\noindent(a) $\Pho(\cdot)$ is affine on $[0,q_0]$ and strictly concave on $(q_0,1]$ for a constant $q_0\in(0,1)$. Specifically, 
	\begin{align}\label{eq:Pho_S}
		\Pho(q)=
		\begin{cases}
			\Phi_+'(q_0)q-1, &\quad 0\le q \le q_0,\\
			\Phi_+(q), &\quad q_0< q \le 1,
		\end{cases}
	\end{align}
	in which $\Phi_+\in \Cc^2\big([q_0,1]\big)$ is a strictly increasing and concave function satisfying $\Phi_+'(q_0)>0$, $\Phi_+''(q_0)<0$, $\Phi_+(q_0)=\Phi_+'(q_0)q_0-1$, $\Phi_+(1)=\Phi_+'(1)=0$, and Assumption \ref{asmp:GrowthCond_S}.
	\\[1ex] 
	\noindent(b) $\Pho(\cdot)$ is strictly concave on $[0,q_0)$ and affine on $[q_0,1]$ for a constant $q_0\in(0,1)$. Specifically,
	\begin{align}\label{eq:Pho_RS}
		\Pho(q)=
		\begin{cases}
			\Phi_-(q), &\quad 0\le q < q_0 ,\\
			\Phi_-'(q_0)(q-1), &\quad q_0\leq  q \le1,
		\end{cases}
	\end{align}
	in which $\Phi_-\in \Cc^2\big((0,q_0]\big)$ is a strictly increasing and concave function satisfying $\Phi_-'(q_0)>0$, $\Phi_-''(q_0)<0$, $\Phi_-(q_0)=\Phi'_-(q_0)(q_0-1)$, $\Phi_-(0)=-1$, $\Phi_-'(0^+)=+\infty$, and Assumption \ref{asmp:GrowthCond_RS}.
	\\[1ex]
	\noindent(c) $\Pho(\cdot)$ is strictly increasing and concave, thrice continuously differentiable on $(0,1)$, $\Pho(0)=-1$, and $\Pho(1)=0$. Furthermore, either Assumption \ref{asmp:GrowthCond_Conc1} holds, or Assumption \ref{asmp:GrowthCond_Conc2} holds.
	\qed
\end{assumption}

Before proceeding further, let us discuss Assumption \ref{asmp:Pho} and, in particular, its relevance to the rank-dependent utility theory (RDUT). Note that $\Pho(\cdot)$ in \eqref{eq:IntEq} represents the concave envelope of $\Phi_n(\cdot,\om)$ in \eqref{eq:Phi} which, in turn, is the counterpart of $\varphi(\cdot)$ in equation (3.1) of \cite{Xu2016}. In other words, $\Pho(\cdot)$ corresponds to $\del(\cdot)$ in equation (4.8) of \cite{Xu2016}. Therein, $\del(\cdot)$ is assumed to be absolutely continuous, strictly increasing, concave (but not necessarily strictly concave), $\del(0)=-1$, and $\del(1)=0$. $\Pho(\cdot)$ satisfies these conditions in all three cases of Assumption \ref{asmp:Pho}. \cite{Xu2016} also lists additional conditions on $\varphi(\cdot)$ which are common in RDUT literature, see Examples 3.1, 3.2, and 3.3 therein.\footnote{Note, however, that the results of \cite{Xu2016} do not rely on these additional conditions.} Specifically, it is commonly assumed that $\varphi(\cdot)$ is either S-shaped (i.e. convex-concave), reverse S-shaped (i.e. concave-convex), or strictly concave. These cases correspond respectively to parts (a), (b), and (c) of Assumption \ref{asmp:Pho}. In particular, \eqref{eq:Pho_S} (respectively, \eqref{eq:Pho_RS}) holds if $\Pho(\cdot)$ is the concave envelope of an S-shaped (respectively, a reverse S-shaped) function, while the strictly concavity of $\Pho(\cdot)$ in Assumption \ref{asmp:Pho}.(c) implies that $\Pho(\cdot)$ is the concave envelope of a strictly concave function.

\begin{remark}\label{rem:Monotonicity}
	The special monotonicity condition of \cite{JinZhou2008} is equivalent to $\Pho(\cdot)$ (and, thus, $\Phi_n(\cdot,\om)$ in \eqref{eq:Phi}) being strictly concave. Therefore, this monotonicity condition is satisfied under Assumption \ref{asmp:Pho}.(c), but not under Assumption \ref{asmp:Pho}.(a) and Assumption \ref{asmp:Pho}.(b). That is, RDPFPPs accommodate for probability distortion functions that do not satisfy the monotonicity condition of \cite{JinZhou2008}. As we discussed in the introduction, this is one characteristic that separates RDPFPs from \emph{forward rank-dependent performance criteria} developed by \cite{HeStrubZariphopoulou2021}, which are only applicable to distortion functions satisfying the monotonicity condition.\qed
\end{remark}

In addition to these common conditions, Assumption \ref{asmp:Pho} imposes further technical conditions on $\Pho(\cdot)$, namely, twice or thrice continuous differentiability, that either $\Pho'(1^-)=0$ or $\Pho'(0^+)=+\infty$, and that the derivatives of $\Pho(q)$ satisfy certain growth conditions  as $q\to0^+$ or $q\to0^-$. These conditions play two important roles. Firstly, they enable us to transform \eqref{eq:IntEq} into a linear Volterra integral equation, which forms the basis of our analysis. Secondly, they guarantee that the resulting Volterra equation has a unique solution and is well-posed.


In the following three subsections, we analyze the integral equation \eqref{eq:IntEq} and provide conditions for existence and uniqueness of its solution. Since the arguments differ depending on which of part of Assumption \ref{asmp:Pho} is holding, we present each case in a separate subsection.

%
%
\subsection{
	$\Pho(\cdot)$ is the concave envelope of an S-shaped function
}\label{subsec:S}
Consider the integral equation \eqref{eq:IntEq} and assume that $\Pho(\cdot)$ satisfies Assumption \ref{asmp:Pho}.(a). As we discussed, this is the case in which we are solving \eqref{eq:IntEq_n} and $\Phi_n(\cdot,\om)$ given by \eqref{eq:Phi} is an S-shaped function.

We start by transforming \eqref{eq:IntEq} into a Volterra integral equation, namely, \eqref{eq:Volterra_S} below. Since $\Pho(\cdot)$ satisfies \eqref{eq:Pho_S}, we have
\begin{align}
	\int_0^1 I\big(y\Pho'(\eta)\big)\Pho'(\eta)\dd \eta
	&= q_0 I\big(y\Phi_+'(q_0)\big)\Phi_+'(q_0)
	+\int_{q_0}^1 I\big(y\Phi_+'(\eta)\big)\Phi_+'(\eta)\dd \eta.
\end{align}
Therefore, the integral equation \eqref{eq:IntEq} becomes
\begin{align}\label{eq:IntEq_S0}
	I\big(y\Phi_+'(q_0)\big)
	+\int_{q_0}^1 I\big(y\Phi_+'(\eta)\big)\frac{\Phi_+'(\eta)}{q_0 \Phi_+'(q_0)}\dd \eta = \frac{I_0(y)}{q_0 \Phi_+'(q_0)},\quad y>0.
\end{align}
The following lemma reduces this integral equation into a linear Volterra integral equation. In its statement, we have defined
\begin{align}\label{eq:J0_S}
	\textstyle
	J_0(t) := \frac{1}{q_0 \Phi_+'(q_0)}I_0\left(\frac{t}{\Phi_+'(q_0)}\right),\quad t>0,
\end{align}
and
\begin{align}\label{eq:k_S}
	\textstyle
	k(t,s) :=
	\frac{s\Phi_+'(q_0)}{q_0 t^2\Phi_+''\left((\Phi_+')^{-1}\left(\frac{\Phi_+'(q_0) s}{t}\right)\right)},\quad 0<s\le t.\qquad
\end{align}

\begin{lemma}\label{lem:Volterra_S}
	Under Assumption \ref{asmp:Pho}.(a), $I(\cdot)$ satisfies \eqref{eq:IntEq_S0} if and only if it satisfies
	\begin{align}\label{eq:Volterra_S}
		I(t) = J_0(t) + \int_0^t I(s)k(t,s)\dd s,\quad t>0,
	\end{align}
	with $J_0(\cdot)$ and $k(\cdot,\cdot)$ given by \eqref{eq:J0_S} and \eqref{eq:k_S}, respectively.\qed
\end{lemma}
\begin{proof}
	Assume that $I(\cdot)$ satisfies \eqref{eq:IntEq_S0}. Since $\Phi_+'(q_0)>0$ by Assumption \ref{asmp:Pho}.(a), we may set $t=y\Phi_+'(q_0)$ in \eqref{eq:IntEq_S0} to obtain
	\begin{align}\label{eq:IntEq_S}
		I(t)
		+\int_{q_0}^1 I\left(\frac{t\Phi_+'(\eta)}{\Phi_+'(q_0)}\right)\frac{\Phi_+'(\eta)}{q_0 \Phi_+'(q_0)}\dd \eta 
		= \frac{1}{q_0 \Phi_+'(q_0)}I_0\left(\frac{t}{\Phi_+'(q_0)}\right),\quad t>0.
	\end{align}
	Since $\Phi_+(\cdot)$ is strictly concave (thus, $\Phi_+'(\cdot)$ is strictly decreasing) by Assumption \ref{asmp:Pho}.(a), we may change variable from $\eta$ to $s=t\Phi_+'(\eta)/\Phi_+'(q_0)$ to obtain
	\begin{align}
		\int_{q_0}^1 I\left(\frac{t\Phi_+'(\eta)}{\Phi_+'(q_0)}\right)\frac{\Phi_+'(\eta)}{q_0 \Phi_+'(q_0)}\dd \eta 
		&=-\int_0^t I(s)\,\frac{s}{q_0 t}\,\frac{\partial}{\partial s}\left[(\Phi_+')^{-1}\left(\frac{\Phi_+'(q_0) s}{t}\right)\right]\dd \eta,
	\end{align}
	in which the lower integral bound on the right side is zero since $\Phi_+'(1)=0$ by Assumption \ref{asmp:Pho}.(a). By substituting the last integral in \eqref{eq:IntEq_S} and noting that $\big((\Phi_+')^{-1})'(x)=1/\Phi_+''\big((\Phi_+')^{-1}(x)\big)$, we conclude that $I(\cdot)$ satisfies \eqref{eq:Volterra_S}. We have thus shown the ``only if'' part of the statement. The ``if'' part of the statement is shown by reversing the argument.
\end{proof}

Thus, to solve the integral equation \eqref{eq:IntEq} under Assumption \ref{asmp:Pho}.(a), we need to solve the Volterra equation \eqref{eq:Volterra_S}. The standard approach for solving the latter equation is to use successive approximation.\footnote{See, for instance, Section 1.3 of \cite{Corduneanu1991}.} There are however some technical difficulties related to the fact that both the solution $I(s)$ and the kernel $k(t,s)$ can be unbounded as $s\to 0^+$. To resolve these issues, we introduce the following growth condition on the derivatives of $\Phi_+(\cdot)$.
\begin{assumption}\label{asmp:GrowthCond_S}
	$\Phi_+(\cdot)$ satisfies the growth condition
	\begin{align}\label{eq:G_S}
		G(t):=\int_0^t g(t,s)\dd s<+\infty, \quad t>0,
	\end{align}
	in which we have defined
	\begin{align}\label{eq:g_S}
		g(t,s) :=
		\frac{1}{s}\sup_{1<\xi\le\frac{t}{s}}\left(
		\textstyle-\frac{\Phi_+'(q_0)}{q_0 \xi^2\Phi_+''\left((\Phi_+')^{-1}\left(\frac{\Phi_+'(q_0) }{\xi}\right)\right)}
		\right),
	\end{align}
	for $0<s\le t$.\qed
\end{assumption}

\begin{remark}
	Note that Assumption \ref{asmp:GrowthCond_S} allows for the limit $g(t,0^+)$ to be unbounded, as long as the integral \eqref{eq:G_S} is bounded. In light of \eqref{eq:g_S}, Assumption \ref{asmp:GrowthCond_S} is thus a condition on the behavior of $\Phi_+''\left((\Phi_+')^{-1}(\xi)\right)$ as $\xi\to+\infty$. Since $\Phi_+'(1^-)=0$ by Assumption \ref{asmp:Pho}.(a), we can also interpret Assumption \ref{asmp:GrowthCond_S} as a condition on the behavior of $\Phi_+'(q)$ and $\Phi_+''(q)$ as $q\to 1^-$. \qed
\end{remark}

The solution of \eqref{eq:Volterra_S} is expressed in terms of the so-called \emph{resolvent kernel}
\begin{align}\label{eq:resolvent_S}
	\ks(t,s) := \sum_{i=1}^\infty k_i(t,s),\quad 0<s\le t,
\end{align}
in which $\{k_i(\cdot,\cdot)\}_{i=1}^\infty$ are the \emph{iterated kernels} recursively given  by
\begin{align}\label{eq:ItrKernels_S}
	\begin{cases}
		k_1(\cdot,\cdot) := k(\cdot,\cdot),\\[1ex]
		k_i(t,s) := \int_s^t k_{i-1}(t,u)k(u,s)\dd u
		,\quad 0< s\le t,\, i\ge2,\qquad
	\end{cases}
\end{align}
with $k(\cdot,\cdot)$ in \eqref{eq:k_S}.
To show that $\ks(\cdot,\cdot)$ is well-defined and continuous, we argue as follows. Take an arbitrary $T>0$.
By \eqref{eq:g_S} and \eqref{eq:k_S}, we have
\begin{align}\label{eq:k_bound_S}
	|k(t,s)|=\frac{|k(t/s,1)|}{s}\le g(T,s),\quad 0<s\le t\le T.
\end{align}
From \eqref{eq:ItrKernels_S} and by induction on $i$, we then obtain the following bound
\begin{align}\label{eq:ItrKernels_approx_S}
	|k_i(t,s)|\le g(T,s)\,\frac{\left(\int_s^t g(T,u)\dd u\right)^{i-1}}{(i-1)!},\quad 0< s\le t\le T, \, i\ge1.
\end{align}
By Weierstrass M-test, it then follows that the series $\sum_{i=1}^\infty k_i(t,s)$ is uniformly convergent on all compact subsets of the set $\{(t,s):0<s\le t\le T\}$. Since $T$ is arbitrarily chosen, $\ks(t,s)$ is well-defined and uniformly continuous on compact subsets of $\{(t,s):0<s\le t\}$. For future reference, we also point out that
\begin{align}\label{eq:resolvent_bound_S}
	|\ks(t,s)|\le g(t,s)\ee^{G(t)},\quad 0<s\le t,
\end{align}
in which $G(\cdot)$ is given by \eqref{eq:G_S}. This inequality follows from
\begin{align}
	\left|\sum_{i=1}^m k_i(t,s)\right|&\le
	\sum_{i=1}^m|k_i(t,s)|\le
	\sum_{i=1}^m g(t,s)\frac{\left(\int_s^t g(t,u)\dd u\right)^{i-1}}{(i-1)!}\\
	\label{eq:resolvent_bound2_S}
	&\le
	g(t,s)\sum_{i=1}^\infty \frac{G(t)^{i-1}}{(i-1)!}\le g(t,s)\ee^{G(t)},
\end{align}
and then letting $m\to\infty$.

We end this subsection by providing the unique solution of the integral equation \eqref{eq:IntEq} under Assumption \ref{asmp:Pho}.(a) and additional conditions on $I_0(\cdot)$, namely, that it is continuous and satisfies the growth condition \eqref{eq:I0_S}. Note that the unique solution, given by \eqref{eq:IntEq_sol_S}, is explicitly expressed in terms of $I_0(\cdot)$ and the resolvent kernel $\ks(\cdot,\cdot)$ of \eqref{eq:resolvent_S}.

\begin{proposition}\label{prop:Volterra_S}
	Assume that $\Pho(\cdot)$ satisfies Assumption \ref{asmp:Pho}.(a), with $\Pho_+(\cdot)$ and $q_0$ in \eqref{eq:Pho_S} and $g(\cdot,\cdot)$ in \eqref{eq:g_S}. Assume further that $I_0:\Rb_+\to\Rb$ is continuous and satisfies
	\begin{align}\label{eq:I0_S}
		\int_0^t|I_0(s)|g(t,s)\dd s < +\infty,\quad t> 0.
	\end{align}
	Then, the unique continuous solution of the integral equation \eqref{eq:IntEq} is given by
	\begin{align}\label{eq:IntEq_sol_S}
		I(y) &= \frac{1}{q_0 \Phi_+'(q_0)}\left[I_0\left(\frac{y}{\Phi_+'(q_0)}\right)
		+ \int_0^y I_0\left(\frac{s}{\Phi_+'(q_0)}\right)\ks(y,s)\dd s\right]\\
		&= \frac{1}{q_0 \Phi_+'(q_0)}\left[I_0\left(\frac{y}{\Phi_+'(q_0)}\right)
		+ \int_0^1 I_0\left(\frac{y\xi}{\Phi_+'(q_0)}\right)\ks(1,\xi)\dd \xi\right],\quad y>0,
	\end{align}
	in which $\ks(\cdot,\cdot)$ is the resolvent kernel in \eqref{eq:resolvent_S}.\qed
\end{proposition}
\begin{proof}
	See Appendix \ref{app:Volterra_S}.
\end{proof}

%
%
\subsection{
	$\Pho(\cdot)$ is the concave envelope of a reverse S-shaped function
}\label{subsec:RS}
Next, we solve \eqref{eq:IntEq} under Assumption \ref{asmp:Pho}.(b), which corresponds to the case in which $\Phi_n(\cdot,\om)$ given by \eqref{eq:Phi} is a reverse S-shaped function.

As in the previous subsection, we first transform the integral equation into a Volterra integral equation. Indeed, if $\Pho(\cdot)$ is of the form \eqref{eq:Pho_RS}, then the integral equation \eqref{eq:IntEq} becomes
\begin{align}\label{eq:IntEq_RS0}
	I\big(y\Phi_-'(q_0)\big)
	+\int_0^{q_0} I\big(y\Phi_-'(\eta)\big)\frac{\Phi_-'(\eta)}{(1-q_0)\Phi_-'(q_0)}\dd \eta
	= \frac{I_0(y)}{(1-q_0)\Phi_-'(q_0)},\quad y>0.
\end{align}
By changing the variable $y$ to $t:=\frac{1}{y\Phi_-'(q_0)}$ and the variable $\eta$ to  $s:=\frac{t\Phi_-'(q_0)}{\Phi_-'(\eta)}$, we obtain a Volterra integral equation for $J(\cdot)=I(1/\cdot)$. We state this result as the following lemma. In its statement, we have defined
\begin{align}\label{eq:J0_RS}
	\textstyle
	J_0(t) :=  \frac{1}{(1-q_0)\Phi_-'(q_0)}I_0\left(\frac{1}{t\Phi_-'(q_0)}\right),\quad t>0,
\end{align}
and
\begin{align}\label{eq:k_RS}
	\textstyle
	k(t,s) &:=
	\frac{t^2\Phi_-'(q_0)}{(1-q_0)s^3\Phi_-''\left((\Phi_-')^{-1}\left(\frac{\Phi_-'(q_0)t }{s}\right)\right)},\quad 0<s\le t.\qquad
\end{align}

\begin{lemma}\label{lem:Volterra_RS}
	Under Assumption \ref{asmp:Pho}.(b), $I(\cdot)$ satisfies \eqref{eq:IntEq_RS0} if and only if $J(t)=I\left(\frac{1}{t}\right)$, $t>0$, satisfies
	\begin{align}\label{eq:Volterra_RS}
		J(t) = J_0(t) + \int_0^t J(s)k(t,s)\dd s,\quad t>0,
	\end{align}
	with $J_0(\cdot)$ and $k(\cdot,\cdot)$ given by \eqref{eq:J0_RS} and \eqref{eq:k_RS}, respectively.\qed
\end{lemma}
\begin{proof}
	The proof is similar to the proof of Lemma \ref{lem:Volterra_S}. We omit the details.
\end{proof}

As in the previous subsection, we have reduced the original integral equation \eqref{eq:IntEq} to a linear Volterra equation, namely, \eqref{eq:Volterra_RS}. To solve this equation, we need the following growth condition on $\Pho_-(\cdot)$.
\begin{assumption}\label{asmp:GrowthCond_RS}
	$\Phi_-(\cdot)$ satisfies the growth condition
	\begin{align}\label{eq:G_RS}
		G(t):=\int_0^t g(t,s)\dd s<+\infty, \quad t>0,
	\end{align}
	in which we have defined
	\begin{align}\label{eq:g_RS}
		g(t,s) :=
		\frac{1}{s}\sup_{1\le\xi\le\frac{t}{s}}\left(
		\textstyle-\frac{\Phi_-'(q_0)\xi^2}{(1-q_0)\Phi_-''\left((\Phi_-')^{-1}\big(\Phi_+'(q_0) \xi\big)\right)}
		\right),
	\end{align}
	for $0<s\le t$.\qed
\end{assumption}
\begin{remark}
	Assumption \ref{asmp:GrowthCond_RS} allows for the limit $g(t,0^+)$ to be unbounded, as long as the integral \eqref{eq:G_RS} is bounded. From \eqref{eq:g_RS}, we can thus interpret Assumption \ref{asmp:GrowthCond_RS} as a condition on the behavior of $\Phi_-''\left((\Phi_-')^{-1}(\xi)\right)$ as $\xi\to+\infty$. Since $\Phi_-'(0^+)=+\infty$ by Assumption \ref{asmp:Pho}.(b), we can also interpret Assumption \ref{asmp:GrowthCond_S} as a growth condition on $\Phi_-'(q)$ and $\Phi_-''(q)$ as $q\to 0^+$. \qed
\end{remark}

The following proposition is the main result of this subsection. It solves the integral equation \eqref{eq:IntEq} under Assumption \ref{asmp:Pho}.(b). The solution is expressed explicitly using the resolvent kernel
\begin{align}\label{eq:resolvent_RS}
	\ks(t,s) := \sum_{i=1}^\infty k_i(t,s),\quad 0<s\le t,
\end{align}
in which $\{k_i(\cdot,\cdot)\}_{i=1}^\infty$ are the iterated kernels
\begin{align}\label{eq:ItrKernels_RS}
	\begin{cases}
		k_1(\cdot,\cdot) := k(\cdot,\cdot),\\[1ex]
		k_i(t,s) := \int_s^t k_{i-1}(t,u)k(u,s)\dd u
		,\quad 0< s\le t,\, i\ge2.\qquad
	\end{cases}
\end{align}
and $k(\cdot,\cdot)$ is given by \eqref{eq:k_RS}. By the argument after \eqref{eq:ItrKernels_S} in the previous subsection, it is shown that $\ks(\cdot,\cdot)$ is well-defined and continuous.
\begin{proposition}\label{prop:Volterra_RS}
	Assume that $\Pho(\cdot)$ satisfies Assumption \ref{asmp:Pho}.(b), with $\Pho_-(\cdot)$ and $g(\cdot,\cdot)$ in \eqref{eq:Pho_RS} and \eqref{eq:g_RS}, respectively. Assume further that $I_0:\Rb_+\to\Rb$ is continuous and satisfies
	\begin{align}\label{eq:I0_RS}
		\int_0^t\left|I_0\left(\frac{1}{s}\right)\right|g(t,s)\dd s < +\infty,\quad t> 0.
	\end{align}
	Then, the unique continuous solution of the integral equation \eqref{eq:IntEq} is given by
	\begin{align}\label{eq:IntEq_sol_RS}
		I(y) = \frac{1}{(1-q_0)\Phi_-'(q_0)}\left[
		I_0\left(\frac{y}{\Phi_-'(q_0)}\right)
		+ \int_0^{\frac{1}{y}} I_0\left(\frac{1}{s\Phi_-'(q_0)}\right) \ks\left(\frac{1}{y},s\right)\dd s
		\right],\quad t>0,
	\end{align}
	in which $\ks(\cdot,\cdot)$ is the resolvent kernel in \eqref{eq:resolvent_RS}.\qed
\end{proposition}
\begin{proof}
	We omit the proof, which is similar to the proof of Proposition \ref{prop:Volterra_S} in the previous subsection.
\end{proof}

%
%
\subsection{
	$\Pho(\cdot)$ is strictly concave
}\label{subsec:Conc}
Finally, we solve the integral equation \eqref{eq:IntEq} when $\Pho(\cdot)$ satisfies Assumption \ref{asmp:Pho}.(c), which corresponds to the case when we are solving the pathwise integral equation \eqref{eq:IntEq_n} with a strictly concave $\Phi_n(\cdot,\om)$. As in the previous two subsections, we solve \eqref{eq:IntEq} by first transforming it into a linear Volterra integral equation. To this end, we need additional conditions on $\Pho(\cdot)$. Specifically, we solve \eqref{eq:IntEq} under either Assumption \ref{asmp:GrowthCond_Conc1} or Assumption \ref{asmp:GrowthCond_Conc2} below. The arguments and the form of the solution differ under each assumption, therefore, we characterize the solution separately in Propositions \ref{prop:Volterra_Conc1} and \ref{prop:Volterra_Conc2}.

We first present the solution of \eqref{eq:IntEq} under Assumption \ref{asmp:Pho}.(c) and with the conditions in Assumption \ref{asmp:GrowthCond_Conc1} below holding. In short, these additional conditions state that $\Pho'(0)$ and $\Pho''(0)$ are bounded and that $\Pho'(1)=0$, while imposing a certain asymptotic behavior on $\Pho'(q)$, $\Pho''(q)$, and $\Pho'''(q)$ as $q\to1^-$. See Remark \ref{rem:GrowthCond_Conc1} below for further discussion on the nature of the asymptotic condition.

\begin{assumption}\label{asmp:GrowthCond_Conc1}
	$\Pho'(1)=0$ and $\max\{\Pho'(0),-\Pho''(0)\}<+\infty$. Furthermore,
	\begin{align}\label{eq:GrowthCond_Conc1}
		G(t):=\int_0^t g(t,s)\dd s<+\infty, \quad t>0,
	\end{align}
	in which we have defined
	\begin{align}\label{eq:kg_Conc1}
		\begin{cases}
			g(t,s) &:=
			\frac{1}{s}\sup_{1<\xi\le\frac{t}{s}} |k(\xi,1)|,\\[1ex]
			k(t,s)&:=
			\textstyle -\frac{\Pho''(0)}{\Pho'(0)^2}\,t\,\frac{\partial^2}{\partial t\partial s}\Pho\left((\Pho')^{-1}\left(\frac{\Pho'(0)\,s}{t}\right)\right),
		\end{cases}		
	\end{align}
	for $0<s\le t$.\qed
\end{assumption}
\begin{remark}\label{rem:GrowthCond_Conc1}
	In Assumption \ref{asmp:GrowthCond_Conc1}, the kernel $k(t,s)$ can be expressed explicitly in terms of $\Pho'(\cdot)$, $\Pho''(\cdot)$, and $\Pho'''(\cdot)$ as follows
	\begin{align}\label{eq:ks_Conc1}
		k(t,s)&:=
		\textstyle
		\frac{\Pho''(0)\Pho'(0) s^2}{t^3\Pho''\left((\Pho')^{-1}\left(\frac{\Pho'(0)s}{t}\right)\right)^3}
		\left[\textstyle
		\frac{2t}{\Pho'(0)s} \Pho''\left((\Pho')^{-1}\left(\frac{\Pho'(0)s}{t}\right)\right)^2
		-\Pho'''\left((\Pho')^{-1}\left(\frac{\Pho'(0)s}{t}\right)\right)
		\right],\qquad
	\end{align}
	for $0<s\le t$. Since \eqref{eq:GrowthCond_Conc1} imposes a growth condition on $g(t,s)$ as $s\to 0^+$ and that $\Pho'(1)=0$ in Assumption \ref{asmp:GrowthCond_Conc1}, we can interpret \eqref{eq:GrowthCond_Conc1} as a condition on asymptotic behavior of $\Pho'(q)$, $\Pho''(q)$, and $\Pho'''(q)$ as $q\to1^-$.\qed
\end{remark}

The following proposition provides the solution of \eqref{eq:IntEq} under Assumptions \ref{asmp:Pho}.(c) and \ref{asmp:GrowthCond_Conc1}. Its statement uses the resolvent kernel $\ks(\cdot,\cdot)$ associated with $k(\cdot,\cdot)$ in \eqref{eq:kg_Conc1}, namely,
\begin{align}\label{eq:Resolvent_Conc1}
	\begin{cases}
		\ks(t,s) := \sum_{i=1}^\infty k_i(t,s),\\[1ex]
		k_1(\cdot,\cdot) := k(\cdot,\cdot),\\[1ex]
		k_i(t,s) := \int_s^t k_{i-1}(t,u)k(u,s)\dd u
		,\qquad
	\end{cases}
\end{align}
for $0<s\le t$. Note that
\begin{align}\label{eq:k_scaling_Conc1}
	k(t,s)=\frac{1}{s}k\left(\frac{t}{s},1\right),\quad 0<s\le t,
\end{align}
by \eqref{eq:ks_Conc1}. By following an argument similar to the one after \eqref{eq:ItrKernels_S} in Subsection \ref{subsec:S}, we can show that $\ks(\cdot,\cdot)$ given by \eqref{eq:Resolvent_Conc1} exists and is continuous.

\begin{proposition}\label{prop:Volterra_Conc1}
	Assume that $\Pho(\cdot)$ satisfies Assumptions \ref{asmp:Pho}.(c) and \ref{asmp:GrowthCond_Conc1}, with $g(\cdot,\cdot)$ given by \eqref{eq:kg_Conc1}. Assume further that $I_0\in \Cc^1(\Rb_+)$ and satisfies
	\begin{align}
		\int_0^t s |I_0'(s)| g(t,s)\dd s < +\infty,\quad t> 0.
	\end{align}
	Then, the unique continuous solution of the integral equation \eqref{eq:IntEq} is given by
	\begin{align}\label{eq:IntEq_sol_Conc1}
		I(y) &= -\frac{\Pho''(0)}{\Pho'(0)^3}\left[
		y\,I_0'\left(\frac{y}{\Pho'(0)}\right)
		+\int_0^y s\,I_0'\left(\frac{s}{\Pho'(0)}\right)\ks(y,s)\dd s
		\right]\\*
		&=-\frac{\Pho''(0)}{\Pho'(0)^3}\left[
		y\,I_0'\left(\frac{y}{\Pho'(0)}\right)
		+\int_0^1 \xi y \,I_0'\left(\frac{\xi y}{\Pho'(0)}\right)\ks(1,\xi )\dd \xi
		\right],\quad y>0,
	\end{align}
	in which $\ks(\cdot,\cdot)$ is the resolvent kernel given by \eqref{eq:Resolvent_Conc1}.\qed
\end{proposition}
\begin{proof}
	See Appendix \ref{app:Volterra_Conc1}.
\end{proof}

Next, we solve \eqref{eq:IntEq} under Assumption \ref{asmp:Pho}.(c) and with the conditions in Assumption \ref{asmp:GrowthCond_Conc2} below holding. In short, these additional conditions state that $\Pho'(1)$ and $\Pho''(1)$ are bounded away from zero, that $\Pho'(0)=+\infty$, and that $\Pho'(q)$, $\Pho''(q)$, and $\Pho'''(q)$ satisfy certain asymptotic behavior as $q\to0^+$. See Remark \ref{rem:GrowthCond_Conc2} below for further discussion on the nature of this asymptotic condition.


\begin{assumption}\label{asmp:GrowthCond_Conc2}
	$\Pho'(0^+)=+\infty$ and $\min\{\Pho'(1),-\Pho''(1)\}>0$. Furthermore, 
	\begin{align}\label{eq:GrowthCond_Conc2}
		G(s):=\int_s^{+\infty} g(t,s)\dd t<+\infty, \quad s>0,
	\end{align}
	in which we have defined
	\begin{align}\label{eq:kg_Conc2}
		\begin{cases}
			g(t,s) &:=
			\frac{1}{t}\sup_{\frac{s}{t}\le \xi<1} |k(1,\xi)|\\[1ex]
			k(t,s)&:= \frac{\Pho''(1)\,s}{\Pho'(1)^2} \frac{\partial^2}{\partial s\partial t}\Pho\left((\Pho')^{-1}\left(\frac{\Pho'(1)\,t}{s}\right)\right),
		\end{cases}		
	\end{align}
	for $0<s\le t$.\qed
\end{assumption}
\begin{remark}\label{rem:GrowthCond_Conc2}
	In terms of $\Pho'(\cdot)$, $\Pho''(\cdot)$, and $\Pho'''(\cdot)$, the kernel $k(\cdot, \cdot)$ in \eqref{eq:kg_Conc2} is expressed as follows
	\begin{align}
		k(t,s)&:=
		\textstyle
		-\frac{\Pho''(1)\Pho'(1) t^2}{s^3\Pho''\left((\Pho')^{-1}\left(\frac{\Pho'(1)t}{s}\right)\right)^3}
		\left[\textstyle
		\frac{2s}{\Pho'(1)t} \Pho''\left((\Pho')^{-1}\left(\frac{\Pho'(1)t}{s}\right)\right)^2
		- \Pho'''\left((\Pho')^{-1}\left(\frac{\Pho'(1)t}{s}\right)\right)
		\right],
	\end{align}
	for $0<s\le t$. Condition \eqref{eq:GrowthCond_Conc2} is essentially a decay condition on $g(t,s)$ as $t\to +\infty$ and $s$ is fixed. Since $\Pho'(0^+)=+\infty$ in Assumption \ref{asmp:GrowthCond_Conc2} and in light of the explicit form of $k(t,s)$ given above, we can interpret \eqref{eq:GrowthCond_Conc2} as a condition on asymptotic behavior of $\Pho'(q)$, $\Pho''(q)$, and $\Pho'''(q)$ as $q\to0^+$.\qed
\end{remark}

We end this subsection (and section) by providing the solution of \eqref{eq:IntEq} under Assumptions \ref{asmp:Pho}.(c) and \ref{asmp:GrowthCond_Conc2}.

\begin{proposition}\label{prop:Volterra_Conc2}
	Assume that $\Pho(\cdot)$ satisfies Assumptions \ref{asmp:Pho}.(c) and \ref{asmp:GrowthCond_Conc2}, with $g(\cdot,\cdot)$ given by \eqref{eq:kg_Conc2}. Assume further that $I_0(\cdot)\in \Cc^1(\Rb_+)$ and satisfies
	\begin{align}
		\int_s^{+\infty} t |I_0'(t)| g(t,s)\dd t < +\infty,\quad s> 0.
	\end{align}
	Then, the unique continuous solution of the integral equation \eqref{eq:IntEq} is given by
	\begin{align}\label{eq:IntEq_sol_Conc2}
		I(y) = \frac{\Pho''(1)}{\Pho'(1)^3}\left[
		y\,I_0'\left(\frac{y}{\Pho'(1)}\right)
		+\int_y^{+\infty} t\,I_0'\left(\frac{t}{\Pho'(1)}\right)\ks(t,y)\dd t
		\right],\quad y>0.
	\end{align}
	in which $\ks(\cdot,\cdot)$ is the resolvent kernel associated with the kernel $k(\cdot,\cdot)$ given by \eqref{eq:kg_Conc2}. That is, $\ks(\cdot,\cdot)$ is given by \eqref{eq:Resolvent_Conc1} and with $k(\cdot,\cdot)$ given by \eqref{eq:kg_Conc2}.\qed
\end{proposition}
\begin{proof}
	The proof is similar to that of the proof of Proposition \ref{prop:Volterra_Conc1}. We omit the details.
\end{proof}

%
%
\section{Completely monotonic RDPFPPs}\label{sec:CMIM}
In Section \ref{sec:Construction}, we proved that constructing RDPFPPs reduces to iteratively solving a pathwise integral equation of the form \eqref{eq:IntEq}, namely,
\begin{align}\label{eq:IntEq_repeat}
	\int_0^1 I\big(y\Pho'(\eta)\big)\Pho'(\eta)\dd\eta = I_0(y),\quad y>0.
\end{align}
In Section \ref{sec:IntEq}, we provided conditions for existence and uniqueness of a continuous solution for this equation. For the iterative construction of RDPFPPs in Section \ref{sec:Construction} to be valid, we additionally need the solution of \eqref{eq:IntEq_repeat} to be an inverse marginal function $I(\cdot)\in\Ic$, with $\Ic$ given by \eqref{eq:Ic}. Finding conditions under which that latter requirement is satisfied is the goal of the current section.

Our main regularity condition is that $I_0(\cdot)$ in \eqref{eq:IntEq_repeat} is a \emph{completely monotonic inverse marginal (CMIM)} function.

\begin{definition}\label{def:CMIM}
	A function $I:\Rb_+\to\Rb_+$ is a completely monotonic inverse marginal (CMIM) function if it is continuously differentiable of all orders, $I(0)=+\infty$, $I(+\infty)=0$, and $(-1)^n\frac{\dd^n}{\dd y^n}I(y)>0$ for all $y>0$ and $n\in\{0,1,\dots\}$. We denote by $\CMIM$ the set of all CMIM functions.\qed
\end{definition}

Clearly, $\CMIM\subset \Ic$, that is, all CMIM functions are inverse marginals. Furthermore, by the celebrated Bernstein's theorem (see \cite{Bernstein1929} or \cite[Theorem 1.4]{SchillingSongVondracek2012}), $I(\cdot)\in\CMIM$ if and only if
\begin{align}\label{eq:CMIM}
	I(y) = \int_0^{+\infty} \ee^{-yz}\mu(\dd z),\quad y>0,
\end{align}
for some non-negative sigma-finite measure $\mu$ on $[0,+\infty)$ satisfying $\mu(\{0\})=0$ and $\mu(\Rb_+)=+\infty$. The last two requirements on $\mu$ follow from Inada's conditions $I'(0)=+\infty$ and $I'(+\infty)=0$, since $I(0^+)=\mu([0,+\infty))$ and $I(+\infty)=\mu(\{0\})$ by \eqref{eq:CMIM}.

As the following examples indicate, $\CMIM$ represents a rich class of (the inverse marginals of) utility functions.

\begin{example}\label{ex:CRRA}
	Inverse marginals of power and logarithmic utilities are CMIM. In fact, if
	\begin{align}
		U(x)=
		\begin{cases}
			\frac{x^{1-\gam}-1}{1-\gam},\quad \gam\in\Rb_+\backslash\{0\},\\[1ex]
			\log(x),\quad \gam=1,
		\end{cases}
	\end{align}
	then the corresponding inverse marginal is given by
	\begin{align}\label{eq:CMIM_CRRA}
		I(y)=\left(U'\right)^{-1}(y)=y^{-\frac{1}{\gam}}=\frac{1}{\Gam\left(\frac{1}{\gam}\right)}\int_0^{\infty} \ee^{-yz} z^{\frac{1}{\gam}-1}\dd z,
	\end{align}
	in which $\Gam(s)=\int_0^{+\infty} t^{s-1}\ee^{-t}\dd t$, $s>0$, is the Gamma function. In particular, for $\gam=1$ (i.e. the case of logarithmic utility), we obtain that $y^{-1}=\int_0^{\infty} \ee^{-yz} \dd z$ for $y>0$.\qed
\end{example}


\begin{example}
	\label{ex:CMIM_Special}
	\cite[Definition 4.6]{Angoshtari23} defines a CMIM function as follows
	\begin{align}\label{eq:CMIM_special}
		I(y) = \int_{\gam_1}^{\gam_2} y^{-\frac{1}{\gam}} m(\dd\gam), \quad y>0,
	\end{align}
	for constants $0<\gam_1<\gam_2$ and a Borel measure $m:\Bc(\Rb_+)\to\Rb_+$ satisfying $m\big([\gam_1,\gam_2]\big)<+\infty$. Note that if $m$ is the Dirac measure concentrated at $1/\gam$, we obtain the CRRA utilities in the previous example. 
	
	As pointed out by Remark 4.7 of \cite{Angoshtari23}, such functions are a special class of (the general) CMIM functions in  Definition \ref{def:CMIM}. This is verified by either directly checking that the requirements of Definition \ref{def:CMIM} are satisfied, or by verifying that $I(\cdot)$ in \eqref{eq:CMIM_special} satisfies the following Bernstein's representation:
	\begin{align}
		I(y)&=\int_{\gam_1}^{\gam_2} y^{-\frac{1}{\gam}} m(\dd\gam)
		=\int_{\gam_1}^{\gam_2} \frac{1}{\Gam(\frac{1}{\gam})}\int_0^\infty\ee^{-yz}z^{\frac{1}{\gam}-1}\dd z\, m(\dd\gam)
		=\int_0^\infty \ee^{-yz}\left(
		\int_{\gam_1}^{\gam_2} \frac{z^{\frac{1}{\gam}-1}}{\Gam(\frac{1}{\gam})} m(\dd \gam)
		\right)\dd z,
	\end{align}
	for $y>0$, in which we used \eqref{eq:CMIM_CRRA} in the second step. In particular, the CMIM functions with representation \eqref{eq:CMIM_special} are precisely those with a Bernstein's representation 
	$I(y)=\int_0^\infty \ee^{-yz} \mu(\dd z)$, $y>0$, in which the measure $\mu$ is absolutely continuous with respect to the Lebesgue measure and has the Radon-Nikodym derivative $\dd \mu/\dd z = \int_{\gam_1}^{\gam_2} \frac{z^{\frac{1}{\gam}-1}}{\Gam(\frac{1}{\gam})} m(\dd \gam)$ for $\gam_1$, $\gam_2$, and $m$ in \eqref{eq:CMIM_special}.\qed
\end{example}

\begin{remark}\label{rem:NotSpecial}
	It is easy to see that not all CMIM functions have the representation \eqref{eq:CMIM_special}. For instance, if the measure $\mu$ is not absolutely continuous with respect to the Lebesgue measure (i.e. that it has an atom or has a singular component), then the corresponding CMIM function given by in \eqref{eq:CMIM} does not have the representation \eqref{eq:CMIM_special}. In other words, our results on CMIM preferences generalize those of \cite{Angoshtari23} in two ways. Firstly, we consider rank-dependent formulation. Secondly, we consider the full class of CMIM preferences (as in Definition \ref{def:CMIM}), while \cite{Angoshtari23} considers only the special subclass of CMIM function in Example \ref{ex:CMIM_Special}. \qed
\end{remark}

The integral equation \eqref{eq:IntEq_repeat} can be easily solved if $I_0(\cdot)$ belongs to the special class of CMIM functions satisfying \eqref{eq:CMIM_special}. Suppose that
\begin{align}\label{eq:I0_CMIM_special}
	I_0(y) = \int_{\gam_1}^{\gam_2} y^{-\frac{1}{\gam}} m_0(\dd \gam),\quad y>0,
\end{align}
for constants $0<\gam_1<\gam_2$ and a Borel measure $m_0$ satisfying $m_0\big([\gam_1,\gam_2]\big)<+\infty$. It is then easy to verify that a solution of \eqref{eq:IntEq_repeat} is
\begin{align}\label{eq:IntEq_sol_CMIM_special}
	I(y) = \int_{\gam_1}^{\gam_2} y^{-\frac{1}{\gam}} \frac{1}{\int_0^1 \Pho'(\eta)^{1-\frac{1}{\gam}}\dd \eta} m_0(\dd \gam),\quad y>0.
\end{align}
Furthermore, if $\Pho(\cdot)$ satisfies Assumption \ref{asmp:Pho}, then the solution of the integral equation must be unique. In particular, the solution given by \eqref{eq:IntEq_sol_CMIM_special} also belongs to the special class of CMIM functions satisfying \eqref{eq:CMIM_special}.

The next theorem is our main result of this section. It states that if $I_0(\cdot)$ is a (general) CMIM function and the existence and uniqueness conditions stated in Section \ref{sec:IntEq} are satisfied, then the unique solution of the integral equation is \eqref{eq:IntEq_repeat} also a CMIM function. In particular, the solution is an inverse marginal.

\begin{theorem}\label{thm:CMIM}
	Assume that $I_0(\cdot)$ and $\Pho(\cdot)$ satisfy all the conditions in at least one of the Propositions \ref{prop:Volterra_S}, \ref{prop:Volterra_RS}, \ref{prop:Volterra_Conc1}, or \ref{prop:Volterra_Conc2}. Assume further that $I_0(\cdot)\in\CMIM$ and that
	\begin{align}\label{eq:CMIM_cond}
		(-1)^n\sum_{j=0}^\infty \frac{(-y)^j}{j!}\left(\int_0^1 \Pho'(\xi)^{j+n+1}\dd \xi\right)^{-1} >0, \quad y>0.
	\end{align}  Then, the integral equation \eqref{eq:IntEq_repeat} has a unique continuous solution $I(\cdot)$ and $I(\cdot)\in\CMIM$.\qed
\end{theorem}
\begin{proof}
	See Appendix \ref{app:CMIM}.
\end{proof}

%
%
\section{RDPFPPs in conditionally complete Black-Scholes market }\label{sec:Example_BS}
In this section, we illustrate our results in the context of the conditionally complete Black-Scholes market of Subsection \ref{subsec:GBM}. Therein, we have verified that Assumptions \ref{asmp:Filtration} and \ref{asmp:rho} are satisfied. By Theorem \ref{thm:RDPFPP} and assuming that $\{W_n(\cdot,\om)\}_{n=1}^\infty$ are (exogenous) probability weighting functions satisfying  Assumption \ref{asmp:ProbDistortion}, we may then construct RDPFPPs by the iterative procedure discussed in Section \ref{sec:Construction}. The main step of this construction is solving the integral equation \eqref{eq:IntEq_n}, that is, to find $I_n(\cdot,\om)\in\Ic$ satisfying
\begin{align}\label{eq:IntEq_n_repeat}
	\int_0^1 I_n\big(y\Pho_n'(\eta,\om),\om\big)\Pho_n'(\eta,\om)\dd \eta=I_{n-1}(y,\om),\quad y>0, 
\end{align}
in which $I_{n-1}(\cdot,\om)\in\Ic$ is known (from the previous iteration) and $\Pho_n(\cdot,\om)$ is the concave envelope of $\Phi_n(\cdot,\om)$ given by \eqref{eq:Phi}, namely,
\begin{align}\label{eq:Phi_repeat}
	\Phi_n(q,\om) 
	=-\int_0^{W_n^{-1}(1-q,\om)} \exp\left(\sqrt{2}\|\lam_n(\om)\|\erf^{\,-1}(2p-1)-\frac{1}{2}\|\lam_n(\om)\|^2\right)
	\dd p,
	\quad 0<q<1,
\end{align}
in which we used \eqref{eq:LogNormQuant}. For the rest of this section, we focus on the integral equation \eqref{eq:IntEq_n_repeat} and properties of its solution.

%
%

\begin{figure}[t]
	\centerline{
			\adjustbox{trim={0.0\width} {0.0\height} {0.0\width} {0.0\height},clip}
			{\includegraphics[scale=0.35, page=1]{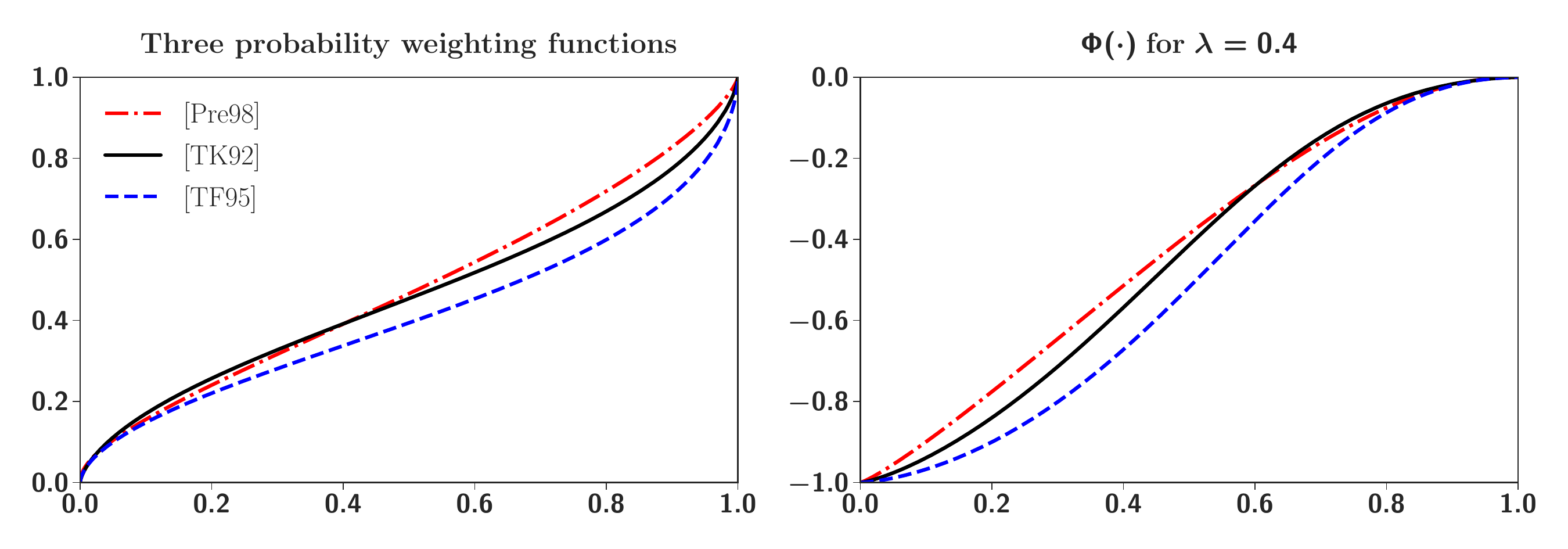}}
	}
	\caption{\textbf{On left:} Plots of three probability weighting functions: $W(p):= p^{0.69} \Big(p^{0.69}+(1-p)^{0.69}\Big)^{-\frac{1}{0.69}}$ from \cite{tversky1992advances} (the solid black curve); $W(p)=0.65 p^{0.6} \Big(0.65 p^{0.6}+(1-p)^{0.6}\Big)^{-1}$ from \cite{tversky1995weighing} (the dashed blue curve); $W(p)=\exp\left(-0.65\big(-\ln(p)\big)^{0.74}\right)$ from \cite{prelec1998probability} (red dashed--dotted curve). Parameter values are taken from various empirical studies, see page 9 of \cite{HeZhou2016}.
		\textbf{On right:} The corresponding plots of the functions $\Phi(\cdot)$ given by \eqref{eq:Phi_BS} for each probability weighting function and assuming the market Sharpe-ratio $\lam=0.4$.
		\label{fig:ProbDist}}
\end{figure}

As we did in Section \ref{sec:IntEq}, let us simplify the notation by dropping dependence on $n$ and $\om$. That is, we are looking for an inverse marginal $I(\cdot)\in\Ic$ satisfying
\begin{align}\label{eq:IntEq_repeat2}
	\int_0^1 I\big(y\Pho'(\eta)\big)\Pho'(\eta)\dd\eta = I_0(y),\quad y>0,
\end{align}
in which $I_0(\cdot)\in\Ic$ is given and $\Pho(\cdot)$ is the concave envelope of 
\begin{align}\label{eq:Phi_BS}
	\Phi(q) = -\int_0^{W^{-1}(1-q)} \exp\left(\sqrt{2}\|\lam\|\erf^{\,-1}(2p-1)-\frac{1}{2}\|\lam\|^2\right)
	\dd p,\quad 0<q<1,
\end{align}
for a probability weighting function $W(\cdot)$. Three common choices for $W(\cdot)$ are depicted in Figure \ref{fig:ProbDist} (see the caption of the Figure details). Figure \ref{fig:ProbDist} also illustrates the corresponding functions $\Phi(\cdot)$ given by \eqref{eq:Phi_BS} and assuming that the market Sharpe-ratio is $\lam=0.4$. $\Phi(\cdot)$ is computed by numerically integrating \eqref{eq:Phi_BS}. For all three probability weighting functions, the functions $\Phi(\cdot)$ are S-shaped and the results of Subsection \ref{subsec:S} are relevant. In particular, the concave envelope of $\Phi(\cdot)$ is of the form \eqref{eq:Pho_S}, namely,
\begin{align}\label{eq:Pho_S_repeat}
	\Pho(q)=
	\begin{cases}
		\Phi'(q_0)q-1, &\quad 0\le q \le q_0,\\
		\Phi(q), &\quad q_0< q \le 1,
	\end{cases}
\end{align}
in which $q_0$ is the unique solution of $\Phi_+'(q_0)q_0-\Phi(q_0)=1$. The left plots in Figure \ref{fig:Pho_ks} illustrate the concave envelope $\Pho(\cdot)$ corresponding to each $\Phi(\cdot)$ in Figure \ref{fig:ProbDist}.

Next, we compute the kernel $k(\cdot,\cdot)$ given by \eqref{eq:k_S}, namely,
\begin{align}\label{eq:k_S_repeat}
	\textstyle
	k(t,s) :=
	\frac{s\Phi'(q_0)}{q_0 t^2\Phi''\left((\Phi')^{-1}\left(\frac{\Phi'(q_0) s}{t}\right)\right)},\quad 0<s\le t.\qquad
\end{align}
Since $k(t,s) = \frac{1}{t} k\left(1,\frac{s}{t}\right)$, we only need to compute $k(1,\xi)$ for $0<\xi\le 1$. The latter function is shown by the dashed curves in the middle plots of Figure \ref{fig:Pho_ks} (as mentioned before, each row corresponds to one probability weighting function in Figure \ref{fig:ProbDist}). We then compute the resolvent kernel $\ks(1,\xi) := \sum_{i=1}^\infty k_i(1,\xi)$, $0<\xi\le 1$, with the iterated kernels $k_i(\cdot,\cdot)$ given by \eqref{eq:ItrKernels_S}, namely,
\begin{align}\label{eq:ItrKernels_S_repeat}
	k_1(1,\cdot) := k(1,\cdot),\quad
	k_i(1,\xi) := \int_\xi^1 k_{i-1}(1,u)k(u,\xi)\dd u
	,\quad 0< \xi\le 1,\, i\ge2.\qquad
\end{align}
As the rightmost plots in Figure \ref{fig:Pho_ks} confirm, $\sup_{0<\xi\le1} |k_i(1,\xi)|$ vanish for sufficiently large $i$. Thus, we obtain a good approximation of the resolvent kernel $\ks(1,\xi) := \sum_{i=1}^\infty k_i(1,\xi)$ by summing only a few iterated kernels. The resolvent kernel $\ks(1,\xi)$, $0<\xi\le1$, is shown by the red curve in the middle plots of Figure \ref{fig:Pho_ks}.

%
%

\begin{figure}[p!]
	\centerline{
			\adjustbox{trim={0.0\width} {0.0\height} {0.0\width} {0.0\height},clip}
			{\includegraphics[scale=0.3, page=1]{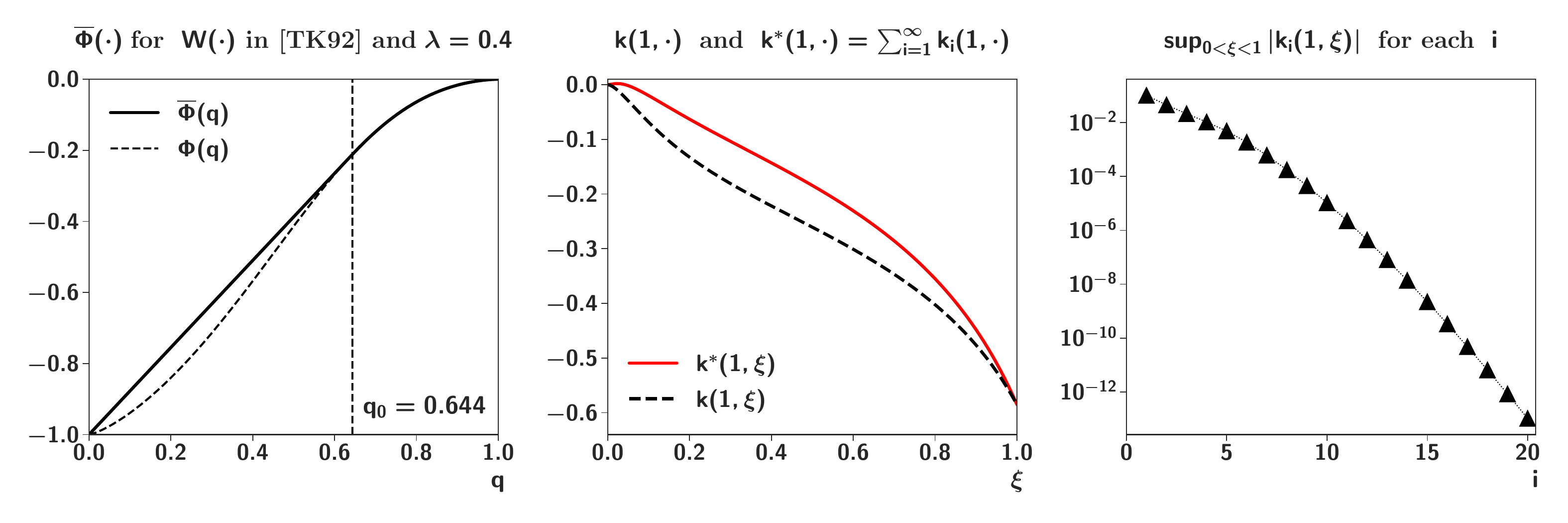}}
	}\vspace{1em}
	\centerline{
			\adjustbox{trim={0.0\width} {0.0\height} {0.0\width} {0.0\height},clip}
			{\includegraphics[scale=0.3, page=1]{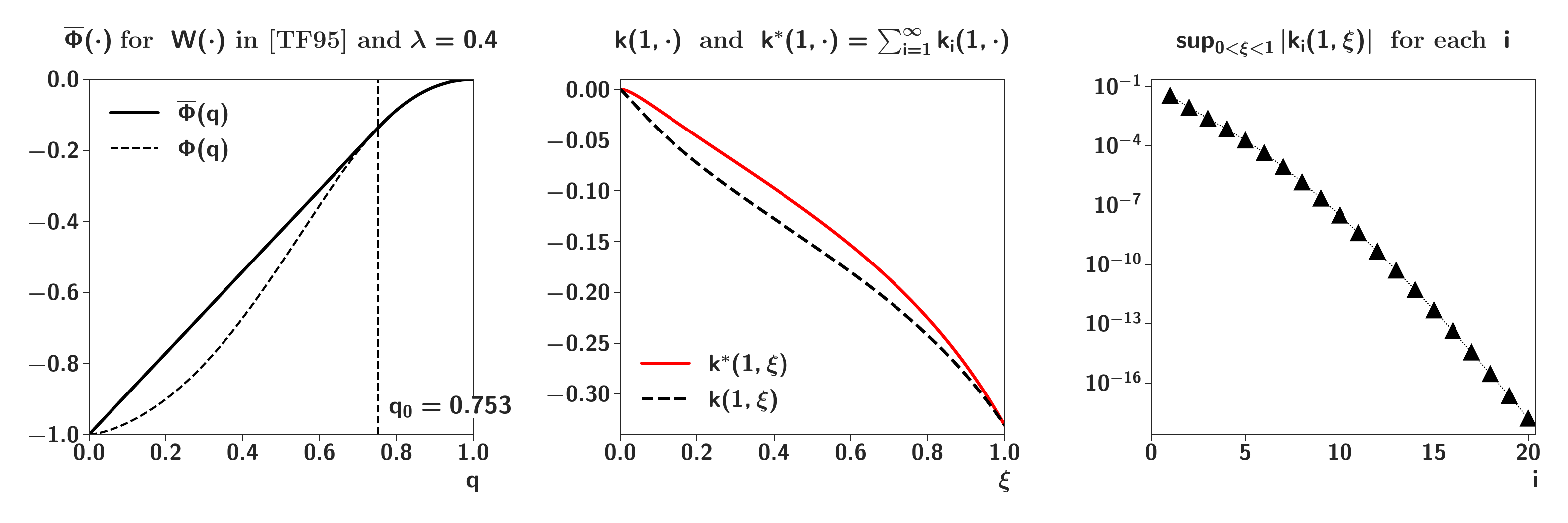}}
	}\vspace{1em}
	\centerline{
			\adjustbox{trim={0.0\width} {0.0\height} {0.0\width} {0.0\height},clip}
			{\includegraphics[scale=0.3, page=1]{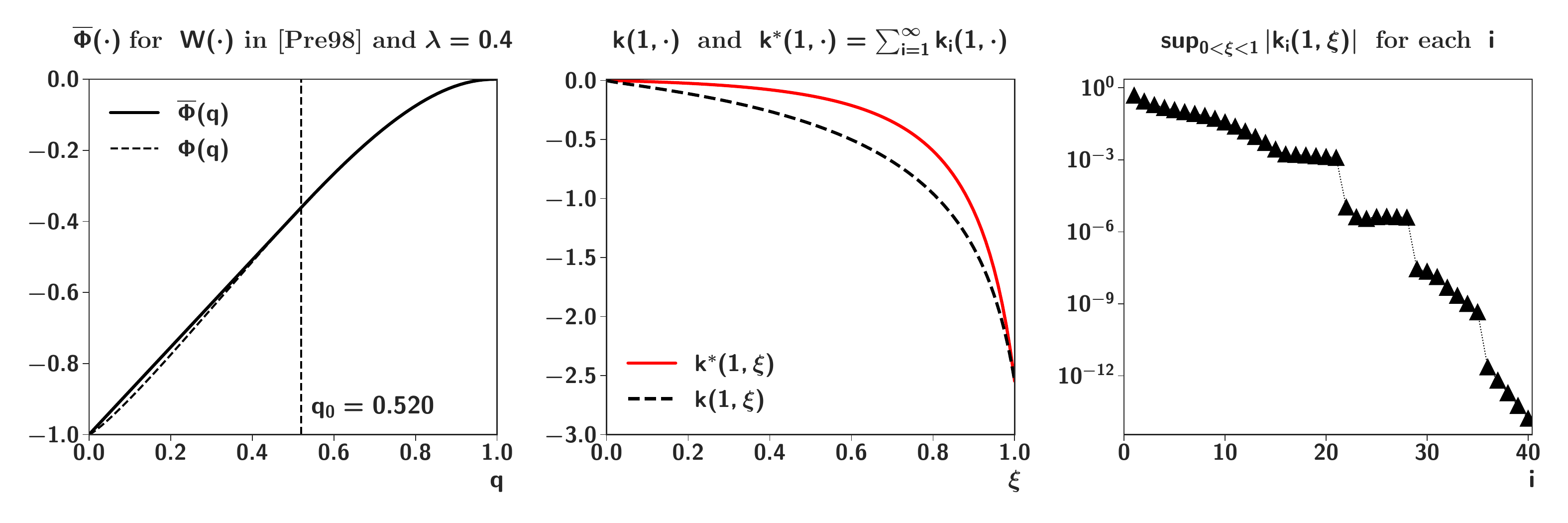}}
	}
	\caption{
		\textbf{Each row} corresponds to one probability weighting function in Figure \ref{fig:ProbDist}. \textbf{Left plots} illustrate the $\Pho(\cdot)$, the concave envelope of each $\Phi(\cdot)$ in Figure \ref{fig:ProbDist}. \textbf{Middle plots} show the corresponding kernel $k(1,\cdot)$ given by \eqref{eq:k_S_repeat} and the corresponding resolvent kernel $\ks(1,\cdot) := \sum_{i=1}^\infty k_i(1,\cdot)$, with the iterated kernels $k_i(1,\cdot)$ given by \eqref{eq:ItrKernels_S_repeat}.
		\textbf{Right plots} show the size of $\sup_{0<\xi\le1} |k_i(1,\xi)|$ for different $i$. Note that the vertical axis has a logarithmic scale. Since $k_i(1,\cdot)$ vanish for sufficiently large $i$, we can approximate $\ks(1,\cdot)$ by summing a finite number of iterated kernels.
		\vspace{2em}
		\label{fig:Pho_ks}}
\end{figure}

%
%

\begin{figure}[p!]
	\centerline{
			\adjustbox{trim={0.0\width} {0.0\height} {0.0\width} {0.0\height},clip}
			{\includegraphics[scale=0.33, page=1]{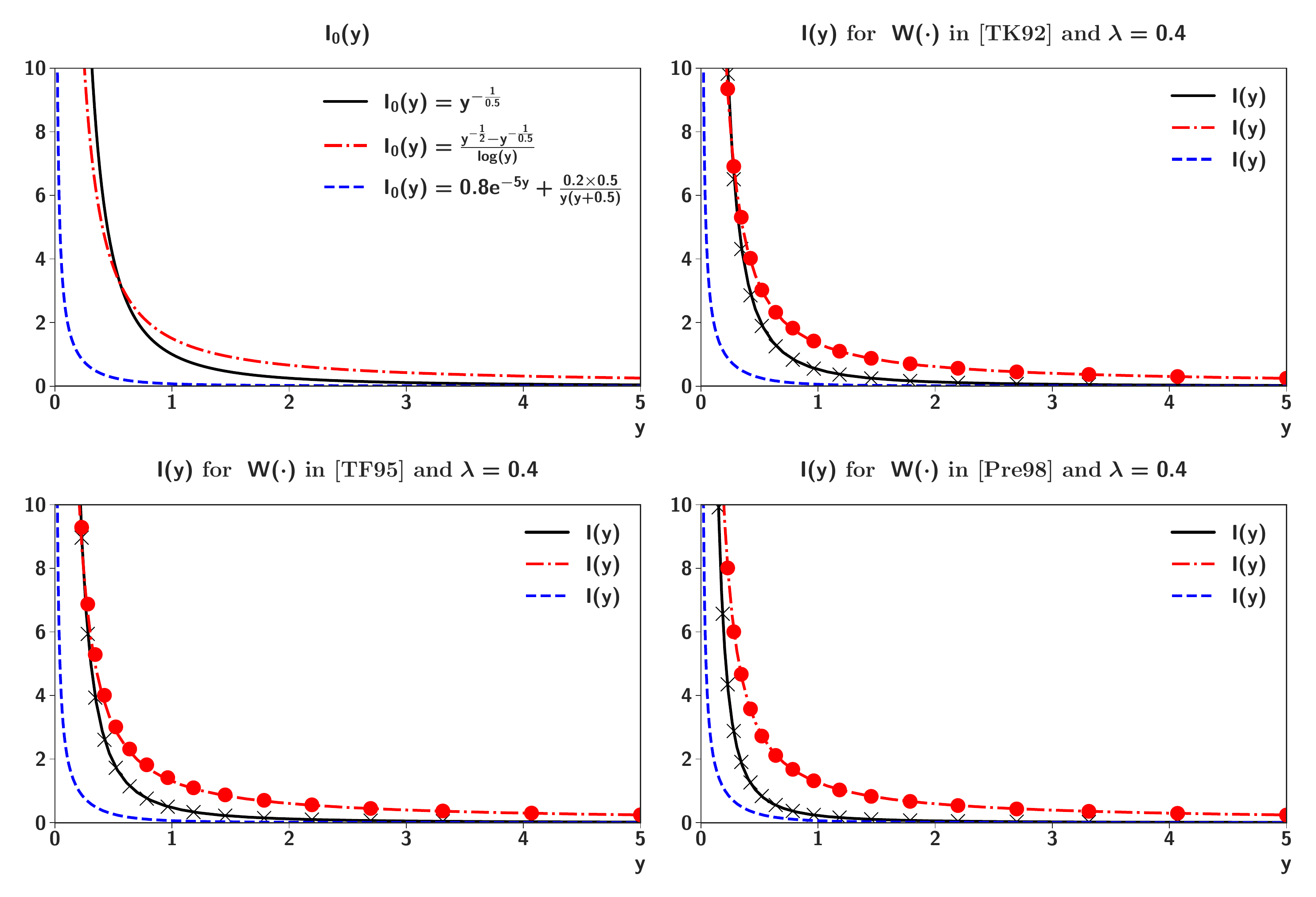}}
	}
	\caption{
		\textbf{Top left:} Plots of three initial inverse marginal functions, namely, $I_0(\cdot)$ given by \eqref{eq:I0_CRRA} (the solid black curve),  $I_0(\cdot)$ given by \eqref{eq:I0_Special} (red dashed-dotted curve), and $I_0(\cdot)$ in \eqref{eq:I0_CMIM} (the blue dashed curve). The first two belong to the special class of CMIM function with representation \eqref{eq:CMIM_special}, while the third one is a CMIM function that does not belong to this class.
		\textbf{Top right:} The solutions $I(y)$ of \eqref{eq:IntEq_repeat2} for each of the three initial inverse marginals, the \cite{tversky1992advances} probability weighting function (see Figure \ref{fig:ProbDist}), and a market Sharpe-ratio $\lam=0.4$. Each curve is computed numerically via \eqref{eq:IntEq_sol_S_repeat2}, with $q_0$, $\Phi'(\cdot)$, and $\ks(1,\cdot)$ in the top row of Figure \ref{fig:Pho_ks}. For first two inverse marginals (namely, the black solid curve and the red dashed-dotted curve), the solution $I(y)$ can be found directly via the closed-form formula \eqref{eq:IntEq_sol_CMIM_special}. These values are shown as data points (namely, black crosses and the red circles), and 		
		match the numerical values calculated via the resolvent formula \eqref{eq:IntEq_sol_S_repeat2}.
		There is no closed-form solution for the third inverse marginal (i.e. the blue dashed curve), as it does not belong the the special class of CMIM functions of the form \eqref{eq:CMIM_special}. Note also that the third curve is less sensitive to the change of the probability weighting function. 
		\textbf{Bottom:} The solutions $I(y)$ of \eqref{eq:IntEq_repeat2} for probability weighting functions in \cite{tversky1995weighing} and \cite{prelec1998probability}.
		\vspace{2em}
		\label{fig:IntEq_sol}}
\end{figure}

With the resolvent kernel $\ks(1,\cdot)$ at hand, we may obtain the solution $I(\cdot)$ to the integral equation \eqref{eq:IntEq_repeat2} by the resolvent formula \eqref{eq:IntEq_sol_S}, namely,
\begin{align}\label{eq:IntEq_sol_S_repeat2}
	I(y) &= \frac{1}{q_0 \Phi'(q_0)}
	\left[I_0\left(\frac{y}{\Phi'(q_0)}\right)
	+ \int_0^1 I_0\left(\frac{y\xi}{\Phi'(q_0)}\right)\ks(1,\xi)\dd \xi\right],\quad y>0.
\end{align}
We use this formula to numerically find $I(\cdot)$ corresponding to three different initial inverse marginals, namely,
\begin{align}\label{eq:I0_CRRA}
	I_0(y) &= y^{-\frac{1}{\gam}}, \quad y>0, \gam=0.5,\\
	\label{eq:I0_Special}
	I_0(y) &=\int_{\gam_1}^{\gam_2} y^{-\frac{1}{\gam}}\frac{1}{\gam^2}\dd \gam = \frac{y^{-\frac{1}{\gam_2}}-y^{-\frac{1}{\gam_1}}}{\log(y)},\quad y>0, \gam_1=0.5, \gam_2=2,
	\intertext{and}
	\label{eq:I0_CMIM}
	I_0(y) &=\beta \ee^{-z_0 y} + (1-\beta)\int_0^{+\infty} \ee^{-yz} (1-\ee^{-\al z})\dd z\\
	&= \beta \ee^{-z_0 y} + (1-\beta) \frac{\al}{y(y + \al)}, \quad y>0, z_0=5, \al=0.5, \beta=0.8.
\end{align}
The graphs of these functions are shown in the top-left plot of Figure \ref{fig:IntEq_sol}.

Note that \eqref{eq:I0_CRRA} is the inverse marginal of the power utility function $U_0(x)=2\sqrt{x}$, and that \eqref{eq:I0_Special} belongs to the special class of CMIM functions with representation \eqref{eq:CMIM_special}. For these two initial inverse marginals, we may use the closed-form formula \eqref{eq:IntEq_sol_CMIM_special} to obtain the solution of the integral equation \eqref{eq:IntEq_sol_S} without using the resolvent formula \eqref{eq:IntEq_sol_S_repeat2}. Thus, we may use these to cases to confirm that the numerical solution obtained by \eqref{eq:IntEq_sol_S_repeat2} is correct.

There is no closed-form solution for the third initial inverse marginal \eqref{eq:I0_CMIM}, since it does not belong to the special class of CMIM function with representation \eqref{eq:CMIM_special}. In particular, this $I_0(y)$ has the Bernstein's representation \eqref{eq:CMIM} with $\mu(\dd z) = 0.8\del_5(\dd z) + (1-\ee^{-0.5 z})\dd z$, in which $\del_5$ is Dirac's measure concentrated at $5$. As we discussed in Remark \ref{rem:NotSpecial}, $I_0(y)$ cannot have the representation \eqref{eq:CMIM_special} since the measure $\mu(\cdot)$ in its Bernstein's representation is atomic.

The top-right plot of Figure \ref{fig:IntEq_sol} show the solutions $I(y)$ of the integral equation \eqref{eq:IntEq_repeat2} corresponding to each of the above three initial inverse marginals, the probability weighting function from \cite{tversky1992advances} (see Figure \ref{fig:ProbDist}), and a market Sharpe-ratio $\lam=0.4$. Each $I(y)$ is computed by numerical integration and formula \eqref{eq:IntEq_sol_S_repeat2} with $q_0$, $\Phi'(\cdot)$ and $\ks(1,\cdot)$ in the top row of Figure \ref{fig:Pho_ks}. For the first two inverse marginals (namely, the black solid curve and the red dashed-dotted curve), we have also shown the value calculated via the closed-form solution \eqref{eq:IntEq_sol_CMIM_special} as data points (namely, black crosses and the red circles). The exact values match with the numerical values calculated via the resolvent formula \eqref{eq:IntEq_sol_S_repeat2}. Note also that the third $I(y)$ (i.e. the blue dashed curve) is almost the same as its initial inverse marginal \eqref{eq:I0_CMIM}. As we already mentioned, there is no closed-form solution for the third $I_0(\cdot)$ in \eqref{eq:I0_CMIM}.

The bottom plots of Figure \ref{fig:IntEq_sol} repeat the experiment using the probability weighting functions in \cite{tversky1995weighing} and \cite{prelec1998probability}. The solutions are slightly different from that of the first probability weighting function, but the plots are described similarly as the top-right plot.

\section{Conclusion and future work}\label{sec:conclude}

We have introduced the notion of rank-dependent predictable forward performance processes that combines rank-dependent utility framework with predictable forward performance processes. We proved that existence of RDPFPPs reduces to solving an integral equation that generalizes the integral equation derived in earlier studies for predictable forward performance processes. We also provided a new approach for solving the integral equations using the theory of Volterra integral equations. For the class of completely monotonic inverse marginal functions, we showed that the solutions obtain by successively solving the integral equation remain within this class.

One future work is to consider a more substantial application of RDPFPP by building on the numerical examples in Section \ref{sec:Example_BS}. In particular, one can consider a specific trading scenario, and study the behavior of optimal wealth generated by different RDPFPPs. As another research direction, one can consider RDPFPPs in a conditionally complete market model with an atomic pricing kernel, building on the analysis in \cite{Xu2014}. As other more challenging future research, one may consider markets that are incomplete (within each trading period), or try to incorporate robust controls.

%
%
\printbibliography 

\appendix

%
%
\section{Proof of Theorem \ref{thm:RDPFPP}}\label{app:RDPFPP}
We need to check that $\{U_n\}_{n=0}^\infty$ satisfies the three conditions in Definition \ref{def:RDPFPP}. Accordingly, we divide the proof into three parts.\vspace{1ex}

\noindent\textbf{Condition $(i)$ of Definition \ref{def:RDPFPP}:} The $\Bc(\Rb_+)\times\Gc_n$-measurability of $U_n(x,\om)$ follows from \eqref{eq:RDPFPP} and the $\Bc((0,1))\times\Gc_n$-measurability of $\Pho_n(p,\om)$ given by \eqref{eq:Pho}. Furthermore, by differentiating \eqref{eq:RDPFPP} with respect to $x$, we obtain that $U_n'(\cdot,\om)= I_n^{-1}(\cdot,\om)$
and, thus, $\left(U_n'(\cdot,\om)\right)^{-1}\in\Ic$ for $\Pb$-almost all $\om$. It then follows that $U_n(\cdot,\om)\in\Uc$ for $\Pb$-almost all $\om$.\vspace{1ex}

\noindent\textbf{Condition $(ii)$ of Definition \ref{def:RDPFPP}:} It suffices to show that
\begin{align}\label{eq:Cond_ii}
	U_{n-1}(x,\om) \ge \int_0^\infty U_{n}(\xi,\om) \dd \Big(1-W_n\big(1-F_{X}(\xi|\Gc_n)(\om),\om\big)\Big),
\end{align}
for all $n\ge1$, $x>0$, $X\in\Ac_n\big(x\big)$, and $\Pb$-almost all $\om$. For a fixed $\om$, \eqref{eq:Cond_ii} holds if
\begin{align}\label{eq:RDProblem}
	\begin{cases}
		U_{n-1}(x,\om) = 
		\displaystyle
		\sup_{X} \int_0^\infty U_{n}(\xi,\om) \dd \Big(1-W_n\big(1-F_{X}(\xi|\Gc_n)(\om),\om\big)\Big),\\[1ex]
		\displaystyle
		\text{subject to: } X \text{ is $\Fc_n$-measurable},\, X\ge0,\,\text{ and } \Eb[\rho_n X|\Fc^+_{n-1}]= x.
	\end{cases}
\end{align}
Assuming that $n$ and $\om$ are fixed, \eqref{eq:RDProblem} can be viewed as a portfolio choice problem under a rank-dependent preference, in which $U_{n-1}(\cdot,\om)$ is the value function and $U_n(\cdot,\om)$ is the terminal utility function. By following the arguments in Sections 2 and 3 of \cite{Xu2016},\footnote{Similar results are also provided by \cite{XiaZhou2016}, among others. Note, however, that our argument uses a change-of-variable introduced by \cite{Xu2016}, see Section 3 therein.} it follows that \eqref{eq:Cond_ii} holds if
\begin{align}
	U_{n-1}(x,\om) &= \int_0^\infty U_{n}(\xi,\om) \dd \Big(1-W_n\left(1-F_{\Xt_{n,x}}(\xi|\Gc_n)(\om),\om\right)\Big)\\*
	\label{eq:Cond_ii_alt1}
	&= \int_0^1 U_{n}\big(Q^*_n(p,x,\om),\om\big)\dd p,
\end{align}
for all $n\ge1$, all $x>0$, and $\Pb$-almost all $\om$, in which
\begin{align}\label{eq:Cond_ii_alt2}
	\Xt_{n,x}(\om) := 
	Q^*_n\Big(1-W_n\big(F_n(\rho_n(\om)|\Gc_n)(\om),\om\big), x, \om\Big),
\end{align}
and $Q^*_n(p,x,\om)$ is the solution of the \emph{quantile optimization problem}\footnote{\eqref{eq:QuantForm} is not the usual quantile formulation of the rank-dependent optimal portfolio problem \eqref{eq:RDProblem}. It is obtained \emph{after} applying a change of variable to the quantile formulation, as discussed in Section 3 of \cite{Xu2016}.}
\begin{align}\label{eq:QuantForm}
	\begin{cases}
		\displaystyle
		\sup_{Q} \int_0^1 U_{n}\big(Q(p),\om\big)\dd p,\\[1ex]
		\displaystyle
		\text{Subject to: }	Q\in\Qc \text{ and } \int_0^1 Q(p)\Phi_n'(p,\om)\dd p = x.
	\end{cases}
\end{align}
In \eqref{eq:QuantForm}, $\Qc$ is the set of all quantile functions of non-negative random variables, namely,
\begin{align}\label{eq:Qc}
	\Qc := \Big\{
	Q:(0,1)\to[0,+\infty), 
	\text{non-decreasing, right-continuous with left limit}
	\Big\}.
\end{align}
To show that Condition $(ii)$ of Definition \ref{def:RDPFPP} is satisfied, we thus need to show that \eqref{eq:Cond_ii_alt1} holds, with $\Xt_{n,x}$ given by \eqref{eq:Cond_ii_alt2} and $Q_n^*$ being the optimizer of \eqref{eq:QuantForm}.

For the remainder of the proof (that is, the second part of the proof), fix an $n\ge1$ and an $\om\in\Om$ such that the $\Pb$-almost surely properties of $I_n$, $I_{n-1}$, $\Phi_n$, and $\Pho_n$ hold.
As we argued in the proof of Condition $(i)$ above, we have $U_{n-1}'(\cdot,\om)=I_{n-1}^{-1}(\cdot,\om)$. Thus, by Condition $(ii)$ of Theorem \ref{thm:RDPFPP}, we have
\begin{align}\label{eq:IntEq_aux1}
	\int_0^1 I_n\big(U_{n-1}'(x,\om)\Pho'_n(p,\om),\om\big)\Pho'_n(p,\om)\dd p = x.	
\end{align}
By using Lemma \ref{lem:Aux_integ_ident} in the Appendix \ref{app:RDPFPP_aux}, we obtain that
\begin{align}\label{eq:IntEq_aux2}
	\int_0^1 I_n\big(U_{n-1}'(x,\om)\Pho'_n(p,\om),\om\big)\Phi'_n(p,\om)\dd p = x.	
\end{align}
Theorem 4.1 of \cite{Xu2016} then yields that the optimizer of \eqref{eq:QuantForm} is
\begin{align}\label{eq:QStar}
	Q_n^*(p,x,\om) = I_n\left(U_{n-1}'(x,\om)\Pho_n'(p,\om),\om\right),\quad p\in(0,1), x>0.
\end{align}
Let $\Xt_n$ be given by \eqref{eq:Cond_ii_alt2}, and define
\begin{align}
	\Ut_{n-1}(x,\om) &:= \int_0^1 U_{n}\big(Q^*_n(p,x,\om),\om\big)\dd p\\*
	\label{eq:Cond_ii_alt3}
	&= \int_0^\infty U_{n}(\xi,\om) \dd \Big(1-W_n\left(1-F_{\Xt_{n,x}}(\xi|\Gc_n)(\om),\om\right)\Big),\quad x>0,
\end{align}
in which the second equality follows from the arguments in Sections 2 and 3 of \cite{Xu2016}. To show \eqref{eq:Cond_ii_alt1} (and thus complete this part of the proof), it only remains to show that $\Ut_{n-1}(\cdot,\om)=U_{n-1}(\cdot,\om)$.

By \eqref{eq:QStar} and \eqref{eq:Cond_ii_alt3}, we have
\begin{align}\label{eq:U_n_Recursion}
	\Ut_{n-1}(x,\om) = \int_0^1 U_{n}\Big(I_n\big(U_{n-1}'(x,\om)\Pho_n'(p,\om),\om\big),\om\Big)\dd p,\quad x>0.
\end{align}
Differentiate \eqref{eq:U_n_Recursion} with respect to $x$ to obtain
\begin{align}
	\Ut_{n-1}'(x,\om)&=\int_0^1 U_n'\Big(I_n\big(U_{n-1}'(x)\Pho_n'(p)\big)\Big)I_n'\big(U_{n-1}'(x)\Pho'_n(p)\big)U_{n-1}''(x)\Pho_n'(p)\dd p\\
	&= U_{n-1}(x)\int_0^1 I_n'\big(U_{n-1}'(x)\Pho_n'(p)\big) \left(\Pho_n'(p)\right)^2U_{n-1}''(x)\dd p\\
	\label{eq:Ut_U_p}
	&=U_{n-1}'(x,\om),
\end{align}
in which we have ignored dependence on $\om$ in the intermediate steps, 
and the last step follows from \eqref{eq:IntEq_aux1}. In particular, differentiating \eqref{eq:IntEq_aux1} with respect to $x$ yields that
\begin{align}
	\int_0^1 I_n'\big(U_{n-1}'(x,\om)\Pho_n'(p,\om),\om\big) \left(\Pho_n'(p,\om)\right)^2U_{n-1}''(x,\om)=1.
\end{align}
Furthermore, \eqref{eq:RDPFPP} and \eqref{eq:Cond_ii_alt3} yield
\begin{align}
	\Ut_{n-1}\big(I_{n-1}(1,\om),\om\big) &= \int_0^1 U_{n}\Big(I_n\big(\Pho_n'(q,\om),\om\big),\om\Big)\dd q\\
	&= \int_0^1
	\left[U_{n-1}\big(I_{n-1}(1,\om),\om\big)+ \int_0^1\int_{I_n\big(\Pho_n'(p,\om),\om\big)}^{I_n\big(\Pho_n'(q,\om),\om\big)} I_n^{-1}(\xi,\om)\dd \xi\dd p\right]\dd q\\
	&= U_{n-1}\big(I_{n-1}(1,\om),\om\big) + \int_0^1
	\int_0^1\int_{I_n\big(\Pho_n'(q,\om),\om\big)}^{I_n\big(\Pho_n'(p,\om),\om\big)} I_n^{-1}(\xi,\om)\dd \xi\dd q\dd p\\
	&= U_{n-1}\big(I_{n-1}(1,\om),\om\big)+\int_0^1\int_0^1\int_1^{I_n\big(\Pho_n'(p,\om),\om\big)} I_n^{-1}(\xi,\om)\dd \xi\dd q\dd p\\
	&\qquad{}-\int_0^1\int_0^1\int_1^{I_n\big(\Pho_n'(q,\om),\om\big)} I_n^{-1}(\xi,\om)\dd \xi\dd q\dd p\\
	\label{eq:Ut_U_1}
	&=U_{n-1}\big(I_{n-1}(1,\om),\om\big).
\end{align}
From \eqref{eq:Ut_U_p} and \eqref{eq:Ut_U_1}, it follows that $\Ut_{n-1}(x,\om)=U_{n-1}(x,\om)$ for all $x>0$. Since the choice of $n$ and $\om$ was arbitrary, we have established validity of \eqref{eq:Cond_ii_alt1}, which concludes the second part of the proof.\vspace{1ex}

\noindent\textbf{Condition $(iii)$ of Definition \ref{def:RDPFPP}:} Take an arbitrary $x_0>0$. We will show that $\{X^*_n\}_{n=0}^\infty$, defined in the statement of Theorem \ref{thm:RDPFPP}, is an optimal portfolio.

First, we check that $\{X^*_n\}_{n=0}^\infty$ is an admissible wealth process. \eqref{eq:Cond_ii_alt2} and \eqref{eq:QStar} yield
\begin{align}\label{eq:Xt_n}
	\Xt_{n,x}(\om) = I_n\Bigg(U_{n-1}'(x,\om)\Pho'_n\bigg(1-W_n\Big(F_n\big(\rho_n(\om)\big|\Gc_n\big)(\om),\om\Big),\om\bigg),\om\Bigg),
\end{align}
for all $n\ge1$, $x>0$, and $\Pb$-almost all $\om$. By \eqref{eq:RDPFPP_optWealth}, we then have $X^*_n(\om) = \Xt_{n,X^*_{n-1}(\om)}(\om)$ for all $n\ge1$ and $\Pb$-almost all $\om$. It follows from Definition \ref{def:Admis} that   $\{X^*_n\}_{n=0}^\infty\in\Ac$ if $\Xt_{n,x}\in\Ac_n(x)$ for all $n\ge1$ and $x>0$. The latter statement is checked as follows. $\Xt_{n,x}$ in \eqref{eq:Xt_n} is $\Fc_n$-measurable since $\rho_n$ is $\Fc_n$-measurable and $I_n$, $U'_{n-1}$, $\Pho_n'$, and $W_n$ are $\Gc_n$-measurable. Furthermore, $\Xt_{n,x}>0$ since $I_n(\cdot)>0$. Finally, that $\Eb[\rho_n \Xt_{n,x}|\Fc^+_{n-1}]=x$ is shown as follows. For $\Pb$-almost all $\om$, by \eqref{eq:Xt_n} and the conditional independence property of $\rho_n$ in Assumption \ref{asmp:rho}, we have
\begin{align}
	&\Eb\left[\rho_n \Xt_{n,x}\middle|\Fc^+_{n-1}\right](\om) = \Eb\left[\rho_n I_n\Bigg(U_{n-1}'(x)\Pho'_n\bigg(1-W_n\Big(F_n\big(\rho_n\big|\Gc_n\big)\Big)\bigg)\Bigg)\middle|\Fc^+_{n-1}\right](\om)\\
	&= \Eb\left[\rho_n I_n\Bigg(U_{n-1}'(x)\Pho'_n\bigg(1-W_n\Big(F_n\big(\rho_n\big|\Gc_n\big)\Big)\bigg)\Bigg)\middle|\Gc_n\right](\om)\\
	&= \int_0^{+\infty} \rho\,I_n\Bigg(U_{n-1}'(x,\om)\Pho'_n\bigg(1-W_n\Big(F_n\big(\rho\big|\Gc_n\big)(\om),\om\Big),\om\bigg),\om\Bigg)\dd \Big(F_n\big(\rho\big|\Gc_n\big)(\om)\Big)\\
	&= \int_0^1
	\,I_n\Bigg(U_{n-1}'(x,\om)\Pho'_n\bigg(1-W_n\Big(q,\om\Big),\om\bigg),\om\Bigg) F^{-1}_{n}(q|\Gc_n)(\om)\dd q\\
	&= \int_0^1
	\,I_n\Bigg(U_{n-1}'(x,\om)\Pho'_n(p,\om),\om\Bigg) F^{-1}_n\big(W_n^{-1}(1-p,\om)\big|\Gc_n\big)(\om)W_n^{(-1)\prime}(1-p,\om)\dd p\\
	&=\int_0^1 I_n\big(U_{n-1}'(x,\om)\Pho'_n(p,\om),\om\big)\Phi'_n(p,\om)\dd p = x,
\end{align}
in which we used \eqref{eq:Phi} for the second to last step, and \eqref{eq:IntEq_aux2} for the last one.

It only remains to show that $\{X^*_n\}_{n=0}^\infty$ is optimal. 
By \eqref{eq:Cond_ii_alt1}, we have
\begin{align}
	U_{n-1}(X^*_{n-1}(\om),\om) &= \int_0^\infty U_{n}(\xi,\om) \dd \left(1-W_n\left(1-F_{\Xt_{n,X^*_{n-1}(\om)}}(\xi|\Gc_n)(\om),\om\right)\right)\\
	&= \int_0^\infty U_n(\xi,\om)\, \dd \big(1-W_n(1-F_{X^*_n}(\xi|\Gc_n)(\om),\om\big),
\end{align}
in which the second equality holds since, as we have already shown, $X^*_n = \Xt_{n,X^*_{n-1}}$.

%
%
\section{An Auxiliary result for the proof of Theorem \ref{thm:RDPFPP}}\label{app:RDPFPP_aux}

We have used the following lemma to obtain \eqref{eq:IntEq_aux2} in the previous Appendix.

\begin{lemma}\label{lem:Aux_integ_ident}
	$\int_0^1 I_n\big(\lam\Pho'_n(p,\om),\om\big)\Phi'(p,\om)\dd p=\int_0^1 I_n\big(\lam\Pho'_n(p,\om),\om\big)\Pho'(p,\om)\dd p$ for all $\lam>0$.
\end{lemma}
\begin{proof}
	Fix an $n\ge1$ and let $\om\in\Om$ be such that the $\Pb$-almost surely properties of $I_n$, $I_{n-1}$, $\Phi_n$, and $\Pho_n$ hold. Let $E=\{p\in(0,1): \Phi_n(p,\om)<\Pho_n(p,\om)\}$. Clearly, $\Phi_n(p,\om)=\Pho_n(p,\om)$ for $p\in(0,1)\backslash E$. Note that if $\Phi_n(p_0,\om)<\Pho_n(p_0,\om)$ for a point $p_0\in(0,1)$, then $\Pho_n$ is linear in a neighborhood of the point $p_0$. Since $\Phi_n(\cdot,\om)$ is absolutely continuous, non-decreasing, and $\Pho_n(\cdot,\om)$ is its concave envelope, there exist countably many constants $a_i,b_i\in(0,1)$ and $c_i\ge0$ such that $0\le a_1<b_1<a_2<b_2,\dots\le 1$, that 
	$E=\cup_i (a_i,b_i)$, and that $\Pho_n'(p,\om)=c_i$ for $p\in (a_i,b_i)$. Note also that $\Phi_n(a_i,\om)=\Pho(a_i,\om)$ and $\Phi_n(b_i,\om)=\Pho_n(b_i,\om)$. Since we have assumed that $b_i<a_{i+1}$, we must have $\Phi_n(p,\om)=\Pho_n(p,\om)$ for $p\in[b_i,a_{i+1}]$. By momentarily ignoring dependence on $n$ and $\om$ (that is, by setting $I_n(y,\om)=I(y)$, $\Phi_n(p,\om)=\Phi(p)$, etc.), we obtain the result as follows
	\begin{align}
		&\int_0^1 I\big(\lam\Pho'(p)\big)\Phi'(p)\dd p
		= \int_{(0,1)\backslash E} I\big(\lam\Pho'(p)\big)\Phi'(p)\dd p 
		+\sum_{i} \int_{a_i}^{b_i} I\big(\lam \Pho'(p)\big)\Phi'(p)\dd p\\
		&= \int_{(0,1)\backslash E} I\big(\lam\Pho'(p)\big)\Pho'(p)\dd p 
		+\sum_{i} I\big(\lam c_i\big) \big(\Phi(b_i)-\Phi(a_i)\big)\\
		&= \int_{(0,1)\backslash E} I\big(\lam\Pho'(p)\big)\Pho'(p)\dd p 
		+\sum_{i} I\big(\lam c_i\big) \big(\Pho(b_i)-\Pho(a_i)\big)
		= \int_0^1 I\big(\lam\Pho'(p)\big)\Pho'(p)\dd p.
	\end{align}
\end{proof}

%
%
\section{Proof of Proposition \ref{prop:Volterra_S}}\label{app:Volterra_S}
Take an arbitrary constant $T>0$ and consider the linear Volterra integral equation
\begin{align}\label{eq:Volterra_S_T}
	J(t) = J_0(t) + \int_0^t J(s)k(t,s)\dd s,\quad 0<t\le T,
\end{align}
with $J_0(\cdot)$ and $k(\cdot,\cdot)$ as in Lemma \ref{lem:Volterra_S}. We will show that the sequence $\{J_m(\cdot)\}_{m=1}^\infty$ given by the successive approximation
\begin{align}\label{eq:SuccAprox_S_T}
	J_m(t) = J_0(t) + \int_{0}^t J_{m-1}(s)k(t,s)\dd s,\quad 0<t\le T, m\ge 1,
\end{align}
is convergent (uniformly on compact subsets of $(0,T]$) to a function $J(\cdot)$ that solves \eqref{eq:Volterra_S_T}.

By induction on $m$, it follows that $\{J_m(\cdot)\}_{m=1}^\infty$ satisfies \eqref{eq:SuccAprox_S_T} if and only if
\begin{align}\label{eq:SuccAprox_itr_S}
	J_m(t) = J_0(t) + \int_0^t J_0(s)\left(\sum_{i=1}^m k_i(t,s)\right)\dd s,\quad 0< t\le T, m\ge1,
\end{align}
in which $\{k_i(\cdot,\cdot)\}_{i=1}^\infty$ are the iterated kernels in \eqref{eq:ItrKernels_S}. Let us first verify that $J_m(\cdot)$ is well-defined and continuous. By \eqref{eq:resolvent_bound2_S}, we have
\begin{align}
	\int_{0}^t \left|J_0(s) \sum_{i=1}^m k_i(t,s)\right| \dd s
	\le \int_0^t |J_0(s)| g(t,s)\ee^{G(t)}\dd s
	\le \ee^{G(T)} \int_0^t |J_0(s)| g(T,s)\dd s,
\end{align}
for any $t\in(0,T]$ and $m\ge1$. Furthermore, $J_0(\cdot)$ (given by \eqref{eq:J0_S}) is continuous since $I_0(\cdot)$ is continuous. Thus, $J_m(\cdot)$ given by \eqref{eq:SuccAprox_itr_S} is well-defined and continuous.

Next, we consider the limit of $J_m(\cdot)$ as $m\to\infty$.  By \eqref{eq:SuccAprox_itr_S}, \eqref{eq:ItrKernels_approx_S}, and \eqref{eq:G_S}, we have
\begin{align}
	|J_m(t)-J_{m-1}(t)| &\le
	\int_0^t |J_0(s) k_m(t,s)|\dd s\\&
	\le\int_0^t |J_0(s)| g(T,s) \,\frac{\left(\int_s^t g(T,u)\dd u\right)^{m-1}}{(m-1)!}\,\dd s\\&
	\label{eq:Cauchy_J_S}
	\le\left(\int_0^T |J_0(s)| g(T,s)\,\dd s\right) \,\frac{G(T)^{m-1}}{(m-1)!},
\end{align}
for all $0<t\le T$ and $m\ge1$. Weierstrass M-test then yields that the series $\sum_{m=1}^\infty \big|J_m(t)-J_{m-1}(t)\big|$ is uniformly convergent for $t$ in compact subsets of $(0,T]$. Therefore, the sequence of functions $\{J_m(\cdot)\}_{m=1}^\infty$ converges (uniformly on compact subsets of $(0,T]$) to a continuous function $J(\cdot)$.
By letting $m\to\infty$ in \eqref{eq:SuccAprox_S_T} and the dominated convergence theorem, it follows that $J(\cdot)$ satisfies \eqref{eq:Volterra_S_T}. Furthermore, letting $m\to\infty$ in \eqref{eq:SuccAprox_S_T} and noting that the resolvent $\ks(\cdot,\cdot)$ is given by \eqref{eq:resolvent_S}, we obtain that
\begin{align}
	J(t) = J_0(t) + \int_0^t J_0(s)\ks(t,s)\dd s,\quad 0< t\le T.
\end{align}
Since the choice of $T$ is arbitrary and does not affect the value of the resolvent kernel $\ks(\cdot,\cdot)$, it follows that $J(t)$ is defined for all $t>0$, that is, $J(\cdot)$ equals to $I(\cdot)$ given by \eqref{eq:IntEq_sol_S}. It then follows that $I(\cdot)$ satisfies \eqref{eq:Volterra_S} and, thus, \eqref{eq:IntEq} by virtue of Lemma \ref{lem:Volterra_S}.

It only remains to show that $I(\cdot)$ is the unique solution of \eqref{eq:IntEq}. Let $\It(\cdot)$ be another continuous solution of \eqref{eq:IntEq}. Take an arbitrary $T>0$. We must have
\begin{align}\label{eq:Jt_S}
	\It(t) = J_0(t) + \int_0^t \It(s)k(t,s)\dd s,\quad 0<s\le t\le T.
\end{align}
Let $\{J_m(\cdot)\}_{m=1}^\infty$ be given by \eqref{eq:SuccAprox_itr_S}. For all $t\in(0,T]$, \eqref{eq:SuccAprox_S_T} yields
\begin{align}
	|\It(t) - J_m(t)|
	&\le \int_0^t |\It(s)-J_{m-1}(t)| |k(t,s)| \dd s\\&\le
	\int_0^t |J_m(s)-J_{m-1}(s)| |k(t,s)| \dd s
	+\int_0^t |\It(s)-J_m(s)| |k(t,s)| \dd s\\&\le
	\left(\int_0^T |J_0(s)| g(T,s)\,\dd s\right) \,\frac{G(T)^m}{(m-1)!}
	+ \int_0^t |\It(s)-J_m(s)| g(T,s) \dd s,
\end{align}
in which we have used \eqref{eq:k_bound_S} and \eqref{eq:Cauchy_J_S} in the last step.
It then follows that
\begin{align}
	|\It(t) - J_m(t)| \le d_m + \int_0^t |\It(s)-J_m(s)| g(T,s) \dd s,\quad 0<t\le T,
\end{align}
in which $d_m:= \left(\int_0^T |J_0(s)| g(T,s)\,\dd s\right) \,\frac{G(T)^m}{(m-1)!}$.
Gronwall's inequality and \eqref{eq:G_S} then yield
\begin{align}
	|\It(t) - J_m(t)| \le d_m + d_m \int_0^t  g(T,s) \ee^{\int_s^t g(T,u)\dd u}\dd s
	\le d_m\left(1+\ee^{G(T)}G(T)\right),\quad 0<t\le T,
\end{align}
It then follows that $\lim_{m\to\infty}J_m = \It(t)$ for all $t\in(0,T]$. However, since $\It(\cdot)$ is continuous and $\{J_m(\cdot)\}_{m=1}^\infty$ converges to $I(\cdot)$ uniformly on compact subsets of $(0,T]$, we must have $\It(t)=I(t)$ for $t\in(0,T]$. Since the choice of $T$ is arbitrary, we obtain that $\It(\cdot)=I(\cdot)$, which shows that $I(\cdot)$ is the unique continuous solution of \eqref{eq:IntEq}.

%
%
\section{Proof of Proposition \ref{prop:Volterra_Conc1}}\label{app:Volterra_Conc1}
Replacing $y$ with $t/\Pho'(0)$ in \eqref{eq:IntEq} yields
\begin{align}
	\int_0^1 I\left(\frac{t\Pho'(\eta)}{\Pho'(0)}\right)\Pho'(\eta)\dd\eta = I_0\left(\frac{t}{\Pho'(0)}\right),\quad t>0.
\end{align}
Since $\Pho(\cdot)$ is strictly concave (thus, $\Pho'(\cdot)$ is strictly decreasing), we may change variable from $\eta$ to $s=t\Pho'(\eta)/\Pho'(0)$ to obtain
\begin{align}\label{eq:IntEq_Volterra1}
	-\int_0^t I(s)\frac{\partial}{\partial s}
	\Pho\left((\Pho')^{-1}\left(\frac{\Pho'(0)\,s}{t}\right)\right)
	\dd s = I_0\left(\frac{t}{\Pho'(0)}\right),\quad t>0,
\end{align}
in which we have used $\Pho'(1)=0$ and
\begin{align}
	\frac{\Pho'(0)\,s}{t}\frac{\partial}{\partial s}\left((\Pho')^{-1}\left(\frac{\Pho'(0)\,s}{t}\right)\right)
	=\frac{\partial}{\partial s}\Pho\left((\Pho')^{-1}\left(\frac{\Pho'(0)\,s}{t}\right)\right).
\end{align}
By differentiating both sides of \eqref{eq:IntEq_Volterra1} with respect to $t$ and using
\begin{align}
	\frac{\partial}{\partial s}\Pho\left((\Pho')^{-1}\left(\frac{\Pho'(0)\,s}{t}\right)\right)
	=\frac{\Pho'(0)^2 s}{t^2}\frac{1}{\Pho''\left((\Pho')^{-1}\left(\frac{s \Pho'(0)}{t}\right)\right)},
\end{align}
we obtain a linear Volterra equation for $I(\cdot)$, namely,
\begin{align}\label{eq:Volterra_Conc1}
	I(t) = J_0(t) + \int_0^t I(s)k(t,s)\dd s,\quad t>0,
\end{align}
in which
\begin{align}
	J_0(t):=-\frac{\Pho''(0)}{\Pho'(0)^3}\,t\,I_0'\left(\frac{t}{\Pho'(0)}\right),\quad t>0,
\end{align}
and $k(\cdot,\cdot)$ is given by \eqref{eq:kg_Conc1}.
The above procedure is reversible. That is, if $I(\cdot)$ satisfies \eqref{eq:Volterra_Conc1}, it must also satisfy the original integral equation \eqref{eq:IntEq}.

Note that
\begin{align}
	|k(t,s)|=\frac{|k(t/s,1)|}{s}\le g(T,s),\quad 0<s\le t\le T,
\end{align}
by \eqref{eq:k_scaling_Conc1} and the definition of $g(\cdot,\cdot)$ in \eqref{eq:kg_Conc1}. We may then follow the argument provided in Appendix \ref{app:Volterra_S} to show that $I(\cdot)$ in \eqref{eq:IntEq_sol_Conc1} is the unique continuous solution of \eqref{eq:Volterra_Conc1} and, thus, that of \eqref{eq:IntEq}. We omit the details of this part of the proof to avoid duplicating arguments.

\section{Proof of Theorem \ref{thm:CMIM}}\label{app:CMIM}
The integral equation has a unique continuous solution since at least one of the stated propositions is applicable. For the rest of the proof, we assume that $I_0(\cdot)$ and $\Pho(\cdot)$ satisfy the conditions in Proposition \ref{prop:Volterra_S}, and prove that the unique solution is a CMIM function. The proofs of the other three cases (that is, conditions in Proposition \ref{prop:Volterra_RS}, Proposition \ref{prop:Volterra_Conc1}, or Proposition \ref{prop:Volterra_Conc2} are satisfied) are similar and, thus, omitted. 

By Proposition \ref{prop:Volterra_S}, the unique solution of \eqref{eq:IntEq_repeat} is given by \eqref{eq:IntEq_sol_S}, namely,
\begin{align}\label{eq:IntEq_sol_S_repeat}
	I(y) = \frac{1}{q_0 \Phi_+'(q_0)}\left[I_0\left(\frac{y}{\Phi_+'(q_0)}\right)
	+ \int_0^1 I_0\left(\frac{y\xi}{\Phi_+'(q_0)}\right)\ks(1,\xi)\dd \xi\right],\quad y>0,
\end{align}
with $q_0$, $\Phi_+(\cdot)$, and $\ks(\cdot,\cdot)$ as in the statement of Proposition \ref{prop:Volterra_S}. Since $I_0(\cdot)\in\CMIM$, Bernstein's theorem yields
\begin{align}\label{eq:CMIM_I0}
	I_0(y) = \int_0^{+\infty} \ee^{-yz}\mu_0(\dd z),\quad y>0,
\end{align}
for some non-negative sigma-finite measure $\mu_0$ on $[0,+\infty)$ satisfying $\mu_0(\{0\})=0$ and $\mu_0(\Rb_+)=+\infty$. Since $I_0(\cdot)$ is continuously differentiable of all orders, we may differentiate with respect to $y$ to obtain
\begin{align}\label{eq:CMIM_I0_dn}
	\frac{\dd^n}{\dd y^n} I_0(y) = \int_0^{+\infty} \ee^{-yz}(-z)^n\mu_0(\dd z),\quad y>0.
\end{align}
In particular, we must have
\begin{align}\label{eq:mu0_integ}
	\int_0^{+\infty} \ee^{-yz}z^n\mu_0(\dd z)<+\infty,\quad y>0, n\ge 0.
\end{align}
From \eqref{eq:IntEq_sol_S_repeat} and \eqref{eq:CMIM_I0}, it follows that
\begin{align}
	I(y)
	&= \frac{1}{q_0 \Phi_+'(q_0)}
	\int_0^{+\infty}\left[
	\exp\left(-\frac{yz}{\Phi_+'(q_0)}\right) + \int_0^1 \exp\left(-\frac{yz\xi}{\Phi_+'(q_0)}\right)\ks(1,\xi)\dd \xi	\right]\mu_0(\dd z),\quad y>0.\qquad.
\end{align}
By differentiating with respect to $y$ and Tonelli's theorem, we obtain that
\begin{align}\label{eq:I_CMIM_prime}
	\frac{\dd^n}{\dd y^n} I(y)
	&=
	\frac{(-1)^n}{q_0 \Phi_+'(q_0)^{n+1}}
	\int_0^{+\infty}\left[
	\exp\left(-\frac{yz}{\Phi_+'(q_0)}\right) + \int_0^1  \exp\left(-\frac{yz\xi}{\Phi_+'(q_0)}\right) \xi^n \ks(1,\xi)\dd \xi	\right] z^n \mu_0(\dd z),\qquad
\end{align}
for $y>0$. By setting $n=1$ and since $\mu(\{0\})=0$ and $\mu(\Rb_+)=+\infty$, we obtain Inada's conditions $I(0)=+\infty$ and $I(+\infty)=0$.

It only remains to show that $(-1)^n\frac{\dd^n}{\dd y^n}I(y)>0$ for all $y>0$ and $n\in\{0,1,\dots\}$. In light of \eqref{eq:I_CMIM_prime}, this condition is satisfied if
\begin{align}\label{eq:LastStep_CMIM}
	\exp\left(-\frac{yz}{\Phi_+'(q_0)}\right) + \int_0^1  \exp\left(-\frac{yz\xi}{\Phi_+'(q_0)}\right) \xi^n \ks(1,\xi)\dd \xi> 0,
	\quad y,n>0. 
\end{align}
Consider a function $\Psi(\cdot)$ given by
\begin{align}\label{eq:Psi}
	\begin{cases}
		\Psi(y):= 1+\sum_{j=1}^\infty \frac{(-y)^j a_j}{j!},\quad y>0,\\[1ex]
		a_j:= \left(\int_0^1 \Pho'(\xi)^{j+1}\dd \xi\right)^{-1}\in(0,1).
	\end{cases}
\end{align}
That $a_j\in(0,1)$ (so that the series in \eqref{eq:Psi} is absolutely convergent) follows from H\"{o}lder's inequality
\begin{align}\label{eq:Pho_bound1}
	\int_0^1\Pho'(\xi)^{j+1}\dd \xi
	> \left(\int_0^1\Pho'(\xi)\dd\xi\right)^{j+1}
	=\big(\Pho(1)-\Pho(0)\big)^{j+1}
	=1,\quad j\ge 1,
\end{align}
and since $\Pho(0)=-1$ and $\Pho(1)=0$ by Assumption \ref{asmp:Pho}. Differentiating \eqref{eq:Psi} yields
\begin{align}\label{eq:Psi_prime}
	(-1)^n\frac{\dd^n}{\dd y^n}\Psi(y)
	&= (-1)^n\sum_{j=n}^\infty \frac{(-y)^{j-n}}{(j-n)!}a_j
	= (-1)^n\sum_{j=0}^\infty \frac{(-y)^j}{j!}a_{j+n}>0,\quad y>0,n\ge0,
\end{align} 
in which the last step follows from \eqref{eq:CMIM_cond}. It is easy to verify that  $\Psi(\cdot)$ solves the integral equation
\begin{align}\label{eq:Psi_IntEq}
	\int_0^1 \Psi\left(y\Pho'(\eta)\right)\Pho'(\eta)\dd\eta = \ee^{-y},\quad y>0.
\end{align}
Indeed,
\begin{align}
	\int_0^1 \Psi\left(y\Pho'(\eta)\right)\Pho'(\eta)\dd\eta
	&=\int_0^1 \sum_{j=0}^\infty \frac{(-y)^j}{j!}\left(\int_0^1 \Pho'(\xi)^{j+1}\dd \xi\right)^{-1} \Pho'(\eta)^{j+1}\dd\eta\\*
	&=\sum_{j=0}^\infty \frac{(-y)^j}{j!}\left(\int_0^1 \Pho'(\xi)^{j+1}\dd \xi\right)^{-1} \int_0^1 \Pho'(\eta)^{j+1}\dd\eta
	=\ee^{-y}.
\end{align}	
By \eqref{eq:G_S}, we have $\int_0^t\ee^{-s}g(t,s)\dd s\le \int_0^tg(t,s)\dd s = G(t)<+\infty$ for $t>0$. By applying Proposition \ref{prop:Volterra_S} to the integral equation \eqref{eq:Psi_IntEq}, we then obtain that
\begin{align}
	\Psi(y)
	&= \frac{1}{q_0 \Phi_+'(q_0)}\left[
	\exp\left(-\frac{y}{\Phi_+'(q_0)}\right) + \int_0^1 \exp\left(-\frac{y\xi}{\Phi_+'(q_0)}\right)\ks(1,\xi)\dd \xi	\right],\quad y>0.\qquad
\end{align}
Therefore,
\begin{align}
	\frac{\dd^n}{\dd y^n}\Psi(y)
	= \frac{(-1)^n}{q_0 \Phi_+'(q_0)^{n+1}}\left[
	\exp\left(-\frac{y}{\Phi_+'(q_0)}\right) + \int_0^1 \exp\left(-\frac{y\xi}{\Phi_+'(q_0)}\right)\xi^n\ks(1,\xi)\dd \xi\right],
	\quad y>0, n\ge0,
\end{align}
Finally, by using \eqref{eq:Psi_prime}, we obtain \eqref{eq:LastStep_CMIM} which concludes the proof.

\end{document}